%% file: 000-main.tex
\documentclass[11pt]{article}
\input{001-preamble}

{\makeatletter
 \gdef\xxxmark{%
   \expandafter\ifx\csname @mpargs\endcsname\relax 
     \expandafter\ifx\csname @captype\endcsname\relax 
       \marginpar{xxx}
     \else
       xxx 
     \fi
   \else
     xxx 
   \fi}
 \gdef\xxx{\@ifnextchar[\xxx@lab\xxx@nolab}
 \long\gdef\xxx@lab[#1]#2{\textcolor{blue}{[\xxxmark #2 ---{\sc #1}]}}
 \long\gdef\xxx@nolab#1{{\bf [\xxxmark #1]}}
}
\usepackage{tikz}
\usepackage{ifthen}
\usetikzlibrary{arrows,shapes}
\usetikzlibrary{intersections,positioning}

\usepackage[margin=1in,letterpaper]{geometry}

\usetikzlibrary{shapes}
\usetikzlibrary{plotmarks}

\newenvironment{proofof}[1]{\noindent{\bf Proof of #1:}}{$\qed$\par}

\title{Spectral Hypergraph Sparsifiers of Nearly Linear Size}

\author{
  Michael Kapralov\thanks{Supported in part by ERC Starting Grant 759471.}\\
  \'Ecole Polytechnique F\'ed\'erale de Lausanne\\
  \texttt{michael.kapralov@epfl.ch}
  \and
  Robert Krauthgamer%
  \thanks{Work partially supported by ONR Award N00014-18-1-2364,
  the Israel Science Foundation grant \#1086/18,
  and a Minerva Foundation grant.
  }\\
  Weizmann Institute of Science\\
  \texttt{robert.krauthgamer@weizmann.ac.il}
  \and
  Jakab Tardos\thanks{Supported by ERC Starting Grant 759471.}\\
  \'Ecole Polytechnique F\'ed\'erale de Lausanne\\
  \texttt{jakab.tardos@epfl.ch}
  \\
  \and
  Yuichi Yoshida\thanks{Supported in part by JSPS KAKENHI Grant Number 18H05291 and 20H05965.}\\
  National Institute of Informatics\\
  \texttt{yyoshida@nii.ac.jp}
}

\begin{document}

\maketitle

\begin{abstract}

Graph sparsification has been studied extensively over the past two decades,
culminating in spectral sparsifiers of optimal size (up to constant factors).
Spectral hypergraph sparsification is a natural analogue of this problem,
for which optimal bounds on the sparsifier size are not known,
mainly because the hypergraph Laplacian is non-linear,
and thus lacks the linear-algebraic structure and tools
that have been so effective for graphs.

Our main contribution is the first algorithm for constructing $\epsilon$-spectral sparsifiers for hypergraphs with $O^*(n)$ hyperedges,
where $O^*$ suppresses $(\epsilon^{-1} \log n)^{O(1)}$ factors.
This bound is independent of the rank $r$ (maximum cardinality of a hyperedge), and is essentially best possible due to a recent bit complexity lower bound of $\Omega(nr)$ for hypergraph sparsification.

This result is obtained by introducing two new tools.
First, we give a new proof of spectral concentration bounds for sparsifiers of graphs;
it avoids linear-algebraic methods,
replacing e.g.~the usual application of the matrix Bernstein inequality
and therefore applies to the (non-linear) hypergraph setting. To achieve the result, we design a new sequence of hypergraph-dependent $\epsilon$-nets on the unit sphere in $\mathbb{R}^n$.
Second, we extend the weight-assignment technique of Chen, Khanna and Nagda [FOCS'20] to the spectral sparsification setting. Surprisingly, the number of spanning trees after the weight assignment can serve as a potential function guiding the reweighting process in the spectral setting.



\end{abstract}

\thispagestyle{empty}
\newpage
\thispagestyle{empty}
\tableofcontents
\setcounter{page}{0}
\newpage

\input{100-intro}

\input{200-prelims}

\input{250-overview}

\input{300-warmup}

\input{400-balanced}
\input{500-mainresult}

\input{600-speedup}


\bibliographystyle{alphaurl}
\bibliography{999-bibliography}

\end{document}

%% file: 001-preamble.tex
\usepackage{algorithm}
\usepackage[noend]{algpseudocode}

\usepackage[margin=1in,letterpaper]{geometry}
\usepackage[T1]{fontenc}
\usepackage{lmodern}
\usepackage{amsmath}
\usepackage{amssymb}
\usepackage{bbm}
\usepackage{amsthm}
\usepackage{color,soul}
\usepackage{xcolor}
\usepackage{graphicx}
\usepackage{bm}
\usepackage{booktabs}
\usepackage{multirow}
\usepackage[export]{adjustbox}
\usepackage{thm-restate}

\usepackage{todonotes}

\usepackage[colorlinks=true, linkcolor=red, urlcolor=blue, citecolor=gray]{hyperref}

\newtheorem{theorem}{Theorem}[section]
\newtheorem{lemma}[theorem]{Lemma}
\newtheorem{definition}[theorem]{Definition}

\newtheorem{corollary}[theorem]{Corollary}
\newtheorem{claim}[theorem]{Claim}
\newtheorem{fact}[theorem]{Fact}

\newtheorem{remark}[theorem]{Remark}

\newcommand{\wt}{\widetilde}

\usepackage{color}
\definecolor{darkgreen}{rgb}{0,0.5,0}
\usepackage{hyperref}
\hypersetup{
	unicode=false,          
	colorlinks=true,        
	linkcolor=red,          
	citecolor=darkgreen,        
	filecolor=magenta,      
	urlcolor=cyan           
}

\usepackage[capitalize, nameinlink]{cleveref}
\crefname{theorem}{Theorem}{Theorems}
\Crefname{lemma}{Lemma}{Lemmas}
\Crefname{claim}{Claim}{Claims}
\Crefname{observation}{Observation}{Observations}

\theoremstyle{plain}

\theoremstyle{plain}
\theoremstyle{plain}

\theoremstyle{plain}

\theoremstyle{plain}

\theoremstyle{definition}

\theoremstyle{definition}

\theoremstyle{definition}

\theoremstyle{plain}
\makeatother

\providecommand{\lemmaname}{Lemma}
\providecommand{\propositionname}{Proposition}
\providecommand{\theoremname}{Theorem}
\providecommand{\corollaryname}{Corollary}
\providecommand{\definitionname}{Definition}
\providecommand{\assumptionname}{Assumption}
\providecommand{\remarkname}{Remark}

\global\long\def\RR{\mathbb{R}}

\global\long\def\poly#1{\mathrm{poly}\left(#1\right)}

\newcommand{\w}{z}

\newcommand{\EX}{{\mathbb E}}

\newcommand{\set}[1]{\{#1\}}

\usepackage{enumitem}
\newlist{UR}{enumerate}{1}
\setlist[UR]{label=2\alph*.}

\allowdisplaybreaks

%% file: 100-intro.tex

\section{Introduction}\label{sec:intro}
We study spectral sparsification of hypergraphs, where the goal is to reduce the size of a hypergraph while preserving its energy.
Given a hypergraph $H=(V,E,w)$ with a weight function $w:E \to \mathbb{R}_+$ over its hyperedges, the \emph{energy} of $x\in\RR^V$ (called a potential vector) is defined as
\[
  Q_H(x) := \sum_{e \in E}w(e)\cdot \max_{u,v \in e}{(x_u - x_v)}^2.
\]
The problem of minimizing $Q_H(x)$ over $x\in \RR^V$ subject to certain constraints appears in many problems involving hypergraphs, including clustering~\cite{Takai2020}, semi-supervised learning~\cite{Hein2013,Yadati2019,Zhang2020} and link prediction~\cite{Yadati2020}, from which we can see the relevance of $Q_H(x)$ in application domains. Note that when $x\in\RR^V$ is a characteristic vector $\mathbbm{1}_S \in \{0,1\}^V$ of a vertex subset $S \subset V$, the energy $Q_H(\mathbbm{1}_S)$ coincides with the total weight of hyperedges cut by $S$, where we say that a hyperedge $e \in E$ is \emph{cut} by $S$ if  $e \cap S\neq \emptyset$ and $e \cap (V \setminus S) \neq \emptyset$. 

Since the number of hyperedges in a hypergraph of $n$ vertices can be $\Omega(2^n)$, it is desirable to reduce the number of hyperedges in the hypergraph while (approximately) preserving the value of $Q_H(x)$ for every $x \in \mathbb{R}^V$, because this lets us speed up any algorithm involving $Q_H$  and reduce its memory usage by running it on the smaller hypergraph instead of $H$ itself.
Soma and Yoshida~\cite{Soma2019} formalized this concept as spectral sparsification for hypergraphs -- a natural generalization of the corresponding concept introduced by the celebrated work of~\cite{Spielman2011} for graphs.
Specifically, for $0<\epsilon<1$, we say that a hypergraph $\widetilde{H}$ is an \emph{$\epsilon$-spectral-sparsifier} of a hypergraph $H$ if $\widetilde{H}$ is a reweighted subgraph of $H$ such that
\begin{equation} \label{eq:HypergraphSparsifier}
  \forall x \in \mathbb{R}^V,
  \qquad
  Q_{\widetilde{H}}(x) \in (1\pm\epsilon) Q_H(x).\footnote{$a \in (1 \pm \epsilon)b$ is a shorthand for $(1-\epsilon)b \leq a \leq (1+\epsilon)b$.}
\end{equation}
We note that when $H$ is an ordinary graph, this definition matches that for graphs~\cite{Spielman2011}.
Soma and Yoshida~\cite{Soma2019} showed that every hypergraph $H$ admits an $\epsilon$-spectral-sparsifier with $\widetilde{O}(n^3/\epsilon^2)$ hyperedges,%
\footnote{Throughout, $\widetilde{O}(\cdot)$ suppresses a factor of $\log^{O(1)} n$.}
and gave a polynomial-time algorithm for constructing such sparsifiers.
Since then the number of hyperedges needed has been reduced to $\widetilde{O}(nr^3/\epsilon^2)$~\cite{Bansal2019},
and recently to $\widetilde{O}(nr/\epsilon^{O(1)})$~\cite{KKTY21},
where $r$ is the maximum size of a hyperedge in the input hypergraph $H$ (called the \emph{rank} of $H$).

The natural question whether every hypergraph admits a spectral sparsifier
with $\wt{O}(n)$ hyperedges (for fixed $\epsilon$) has proved to be challenging.
On the one hand, it is well-known that
a hypergraph is a strictly richer object than an ordinary graph
(hyperedges cannot be ``simulated'' by edges, even approximately),
and in all previous results and techniques,
this extra complication introduced an extra factor of at least $r$. 
On the other hand, an exciting recent result~\cite{Chen2020}
has achieved sparsifiers with $\wt{O}(n)$ hyperedges,
if one is only interested in preserving the \emph{hypergraph cut function},
i.e., satisfying~\eqref{eq:HypergraphSparsifier}
only for all characteristic vectors $x=\mathbbm{1}_S$ where $S \subseteq V$.
Nevertheless, the spectral version of this question has remained open,
primarily due to the non-linearity of the hypergraph Laplacian
and the lack of linear-algebraic tools that have been so effective for graphs.

We settle this question by showing that a nearly linear number of hyperedges
suffices. 

\begin{restatable}{thm}{thmmain}\label{thm:main}
  For every hypergraph with $n$ vertices and every $1/n\le\epsilon\le1/2$, there exists an $\epsilon$-spectral-sparsifier with $O(n\epsilon^{-4}\log^3 n)$ hyperedges.
  Moreover, one can construct such a sparsifier in time $\widetilde{O}(mr + \poly{n})$, where $m$ is the number of hyperedges and $r$ is the maximum size of a hyperedge in $H$.
\end{restatable}
We note that the bit complexity of our sparsifier is tight up to a polylogarithmic factor for a large range of $r$ due to the lower bound of~\cite{KKTY21}.

\subsection{Additional Related Work}


Recall that we call $\widetilde{H} = (V, \widetilde{E}, \widetilde{w})$ an \emph{$\epsilon$-cut sparsifier} of $H = (V, E, w)$ if every cut weight is preserved to within a factor of $1 \pm \epsilon$. 
This definition matches the one for ordinary graphs introduced by Bencz{\'{u}}r and Karger~\cite{Benczur2015}, who showed that every graph has an $\epsilon$-cut-sparsifier with $O(n\log n/\epsilon^2)$ edges, where $n$ is the number of vertices.
For hypergraphs, Kogan and Krauthgamer~\cite{Kogan2015} gave the first construction of non-trivial cut sparsifiers, which uses $O(n(r + \log n)/\epsilon^2)$ hyperedges, where $r$ is the maximum size of a hyperedge.
They also mentioned that the results of Newman and Rabinovich~\cite{Newman2013} implicitly give an $\epsilon$-cut sparsifier with $O(n^2/\epsilon^2)$ hyperedges.
Chen, Khanna, and Nagda~\cite{Chen2020} improved this bound to $O(n \log n/\epsilon^2)$, which is almost tight because one needs $\Omega(n/\epsilon^2)$ edges even for ordinary graphs~\cite{ACKQWZ16,CKST19}.

Spielman and Teng~\cite{Spielman2011} introduced the notion of a spectral sparsifier for ordinary graphs and showed that every graph on $n$ vertices admits an $\epsilon$-spectral sparsifier with $O(n\log^{O(1)} n/\epsilon^2)$ edges.
This bound was later improved to $O(n/\epsilon^2)$~\cite{BatsonSS12}, which is tight~\cite{ACKQWZ16,CKST19}.
The literature on graph sparsification is too vast to cover here,
including~\cite{Spielman2011,Spielman2011a,BatsonSS12,ZhuLO15,LeeS15a,Lee017} and many other constructions,
and we refer the reader to the surveys~\cite{Vishnoi2013,Teng2016}.

For an ordinary graph $G = (V, E, w)$, the \emph{Laplacian} of $G$ is the matrix $L_G = D_G - A_G$, where $D_G \in \mathbb{R}^{V \times V}$ is the diagonal (weighted) degree matrix and $A_G \in \mathbb{R}^{V \times V}$ is the adjacency matrix of $G$.
Then, the energy $Q_G$, defined in~\eqref{eq:HypergraphSparsifier}, can be written also as
\[
  Q_G(x) = x^\top L_G x.
\]
For a hypergraph $H = (V, E, w)$, it is known that we can define a (multi-valued) Laplacian operator $L_H : \mathbb{R}^V \to 2^{\mathbb{R}^V}$, so that
\[
  Q_H(x) = x^\top y
\]
for every $x \in \mathbb{R}^V$ and $y \in L_H(x)$~\cite{Louis2015,Chan2018,Yoshida2019} (hence we can write $Q_H(x)$ also as $x^\top L_H(x)$ without ambiguity).
Although the Laplacian operator $L_H$ is no longer a linear operator, its mathematical property has been actively investigated~\cite{ikeda2019finding,fujii2018polynomialtime,ikeda2021coarse} through the theory of monotone operators and evolution equations~\cite{komura1967nonlinear,miyadera1992nonlinear}.

Yoshida~\cite{Yoshida2016} proposed a Laplacian operator for directed graphs and used it to study structures of real-world networks.
The Laplacian operators for graphs, hypergraphs, and directed graphs mentioned above were later unified and generalized as Laplacian operator for submodular transformations/submodular hypergraphs~\cite{Li2018,Yoshida2019}.



%% file: 200-prelims.tex

\section{Preliminaries}\label{sec:prelims}

In this paper, we deal with spectral sparsification of hypergraphs. For the sake of generality, we consider weighted hypergraphs denoted $H=(V,E,w)$, where $V$ is the vertex set of size $n$, $E$ is the hyperedge set of size $m$, and $w:E\to\mathbb R_+$ is the set of hyperedge weights.
We will also, however, deal with ordinary graphs, that is graphs where each edge contains two vertices exactly. In order to distinguish clearly between graphs and hypergraphs, we will typically denote graphs as $G=(V,F,\w)$, where $V$ is the vertex set, $F$ is the edge set, and $\w:F\to\mathbb R_+$ is the set of edge weights. In general we will use $f$ and $g$ to denote ordinary edges, while reserving $e$ to denote hyperedges.

For simplicity \emph{all graphs and hypergraphs we consider in this paper will be connected}.

\subsection{Spectral Graph Theory}

\begin{definition}\label{def:Laplacian}
	The Laplacian of a weighted graph $G=(V,F,\w)$ is defined as the matrix $L_G\in\mathbb R^{V\times V}$ such that
	$$\left(L_G\right)_{uv}=\begin{cases}
		d(u)&\text{if $u=v$,}\\
		-\w(u,v)&\text{if $(u,v)\in F$,}\\
		0&\text{otherwise.}
		\end{cases}$$
	Here $d(u)$ denotes the weighted degree of $u$, that is the sum of all weights of incident edges.
	Thus $L_G$ is a positive semidefinite matrix, and its quadratic form can be written as
	$$x^\top L_G x=\sum_{(u,v)\in F}\w(u,v)\cdot(x_u-x_v)^2.$$
\end{definition}
The spectral sparsifier of $G$ is defined as a reweighted subgraph which closely approximates the quadratic form of the Laplacian on every possible vector.

\begin{definition}\label{def:sparsifier}
	Let $G=(V,F,\w)$ be a weighted ordinary graph. Let $\wt G=(V,\wt F,\wt{\w})$ be a reweighted subgraph of $G$, defined by $\wt{\w}:F\to\mathbb R_+$, where $\wt F=\{f\in F \mid \wt{\w}(f)>0\}$. For $\epsilon>0$, $\wt G$ is an $\epsilon$-spectral sparsifier of $G$ if for every $x\in\mathbb R^V$
	$$x^\top L_{\wt G}x \in (1 \pm \epsilon)\cdot x^\top L_G x.$$
\end{definition}

The quadratic form of the graph Laplacian from Definition~\ref{def:Laplacian} can be generalized to hypergraphs. Although this generalization is highly non-linear, we still refer to it as the ``quadratic form'' of the hypergraph.

\begin{definition}
	The quadratic form (or sometimes \emph{energy}) of a hypergraph $H=(V,E,w)$ is defined on the input vector $x\in\mathbb R^V$ as
	$$Q_H(x)=\sum_{e\in E}w(e)\cdot\max_{u,v\in e}(x_u-x_v)^2.$$
\end{definition}

Consequently, we may also define the concept of spectral sparsification in hypergraphs, analogously to Definition~\ref{def:sparsifier}:

\begin{definition}\label{def:hyper-sparsifier}
	Let $H=(V,E,w)$ be a weighted hypergraph. Let $\wt H=(V,\wt E,\wt{w})$ be a reweighted subgraph of $H$, defined by $\wt{w}:E\to\mathbb R_+$, where $\wt E=\{e\in E \mid \wt{w}(e)>0\}$.
	For $\epsilon>0$, $\wt H$ is an $\epsilon$-spectral sparsifier of $H$ if for every $x\in\mathbb R^V$
	$$Q_{\wt H}(x) \in (1 \pm \epsilon)\cdot Q_H(x).$$
\end{definition}

\subsection{Effective Resistance}

\begin{definition}\label{def:eff-res}
	Let $G=(V,F,\w)$ be a weighted ordinary graph. The effective resistance of a pair of vertices $(u,v)$ is defined as
	$$R_G(u,v)=(\chi_u-\chi_v)^\top L_G^+(\chi_u-\chi_v).$$
	Here $\chi_u\in\mathbb R^V$ is the vector with all zeros, and a single $1$ at the coordinate corresponding to $u$. $L_G^+$ is the Moore-Penrose pseudo-inverse of $L_G$, which is positive semidefinite.

	We may write $R(u,v)$ in cases where $G$ is clear from context.
\end{definition}

We will often use the notation $R_G(f)=R_G(u,v)$ where $f=(u,v)$ is an edge. It is important to note, however, that effective resistance is a function of the vertex pair, not the edge, and does not depend directly on the weight of $f$.

We now state several well-known and useful facts about effective resistance.

\begin{fact}\label{fact:eff-res}
	The effective resistance of an edge $(u,v)$ is alternatively defined as
	$$R_G(u,v)=\max_{x\in\mathbb R^V}\frac{(x_u-x_v)^2}{x^\top Lx}.$$
\end{fact}

\begin{fact}\label{fact:eff-res-metric}
	Effective resistance constitutes a metric on $V$.
\end{fact}

\begin{fact}\label{fact:w-R}
	For any weighted graph $G=(V,F,\w)$ and any edge $f\in F$ we have $\w(f)\cdot R_G(f)\le1$, with equality if and only if $f$ is a bridge.
\end{fact}

\begin{fact}\label{fact:eff-res-sum}
	For any weighted graph $G=(V,F,\w)$ we have
	$$\sum_{f\in F}\w(f)\cdot R_G(f)=n-1.$$
\end{fact}

\subsection{Chernoff Bound}

\begin{theorem}[Chernoff bound, see for example~\cite{DBLP:books/daglib/0021015}]\label{thm:chernoff}

Let $Z_1,Z_2,\ldots,Z_k$ be independent random variables in the range $[0,a]$. Furthermore, let $\sum Z_i=Z$ and let $\mu\ge\mathbb E(Z)$. Then for any $\delta\in(0,1)$,
$$\mathbb P\left(\left|Z-\mathbb E(Z)\right|\ge\delta\mu\right)\le2\exp\left(-\frac{\delta^2\mu}{3a}\right).$$

\end{theorem}


%% file: 250-overview.tex

\section{Technical Overview}

\subsection{Analyzing Ordinary Graphs}

The sparsification of ordinary graphs is a highly studied topic, with several techniques proposed for the construction of spectral sparsifiers throughout the years~\cite{Spielman2011,Spielman2011a,BatsonSS12,ZhuLO15,LeeS15a,Lee017}. However, the analysis of spectral sparsifiers always relies heavily on the linear nature of the graph Laplacian, e.g., using matrix concentration results such as matrix Bernstein~\cite{Tropp2011} or the work of~\cite{RudelsonV07}. This presents a significant problem when attempting to generalize these techniques to the highly non-linear setting of hypergraph spectral sparsification. Indeed, all previous results lose at least a factor of $r$ due to this obstacle. We therefore dedicate the entirety of our first technical section (Section~\ref{sec:warm-up}) to presenting a new proof of the existence of nearly linear spectral sparsifiers for ordinary graphs. We use the algorithm from~\cite{Spielman2011a}, which constructs a sparsifier $\widetilde{G}$ by sampling each edge with probability proportional to its effective resistance.
However, our proof avoids using matrix concentration inequalities, and instead relies on a more direct chaining technique for proving the concentration of $x^\top L_{\wt G}x$ around its expectation, i.e. $x^\top L_G x$, for all $x$ simultaneously.
To our knowledge, this is the first nearly-optimal direct analysis of spectral sparsification through effective resistance sampling. It will also be the basis of our main result, as we adapt it to the hypergraph setting in Sections~\ref{sec:balanced} and~\ref{sec:main}.

More formally, for an input graph $G=(V,F,\w)$, we define $\wt G$ as the result of sampling each edge $f$ of $G$ independently with probability $p(f)\approx \w(f)\cdot R_G(f)$, and setting its weight to $\wt\w(f)=\w(f)/p(f)$. Our aim is then to prove
\begin{equation}\label{eq:tech-overview-ordinary-main}
x^\top L_{\wt G}x\approx x^\top L_G x
\end{equation}
simultaneously for all $x\in \mathbb R^V$. For simplicity we assume that $x^\top L_G x=1$. Equation~\eqref{eq:tech-overview-ordinary-main} is in fact the concentration of a random variable around its expectation, and so we can use Chernoff bound to prove it for any specific $x$.
Our plan is then to use a combination of Chernoff and union bounds to prove it for all possible $x$. Since $x$ can take any value in $\mathbb R^V$ we must discretize it to some $\epsilon$-net while retaining a good approximation to its quadratic form, i.e. $x^\top L_G x$.

Let us take a closer look at the application of Chernoff bound to Equation~\eqref{eq:tech-overview-ordinary-main}: $x^\top L_{\wt G}x$ is the sum of the independent random variables $\wt\w(u,v)\cdot(x_u-x_v)^2$ for $(u,v) \in F$; hence, by Theorem~\ref{thm:chernoff},
the strength of the bound depends crucially on the upper bound $a$ on values that each random individual random variable can possibly attain. The maximum value of $\wt\w(u,v)\cdot(x_u-x_v)^2$ is attained when $(u,v)$ is sampled in $\wt G$, in which case it is $\approx(x_u-x_v)^2/R_G(u,v)$. Thus
$$\mathbb P\left(x^\top L_{\wt G}x\not\approx x^\top L_G x\right)\lessapprox\exp\left(-\frac1{\max_{(u,v)\in F}(x_u-x_v)^2/R_G(u,v)}\right).$$
This upper bound can be as bad as $\exp(-\wt O(1))$ and is far too crude for our purposes---no sufficiently sparse rounding scheme (i.e., discretization) exists for $x$. We turn to the technique of chaining---the use of progressively finer and finer rounding schemes.

As seen above, the strength of our Chernoff bound depends primarily on the quantity $(x_u-x_v)^2/R_G(u,v)$ for each edge $(u,v)$, which we call the ``power'' of the edge.
Therefore, it makes sense to partition the edges of $G$ into a logarithmic number of classes based on their power, that is $F_i$ contains edges $(u,v)$ for which $(x_u-x_v)^2\approx2^{-i}\cdot R_G(u,v)$.
When focusing only on the subgraphs $G(F_i)$ induced by $F_i$, we get the more fine-tuned Chernoff bound
$$\mathbb P\left(x^\top L_{\wt G(F_i)}x\not\approx x^\top L_{G(F_i)}x\right)\lessapprox\exp\left(-\frac1{\max_{(u,v)\in F_i}(x_u-x_v)^2/R_G(u,v)}\right)\lessapprox\exp\left(-2^i\right).$$

We thus have the task of proposing a rounding scheme $\varphi_i:\mathbb R^V\to\mathbb R^V$ specially for each class $F_i$ such that
\begin{itemize}
	\item the image of $\varphi_i$ is a finite set of size at most $\approx\exp\left(2^i\right)$,
	\item the rounding approximately preserves the quantity $(x_u-x_v)^2$ for $(u,v)\in F_i$.
\end{itemize}

To gain more intuition on what such a rounding scheme must look like, we draw inspiration from the idea of resistive embedding from~\cite{Spielman2011a}. We map the edges in $F_i$, as well as our potential vector $x$, into vectors in $\mathbb R^n$ in such a way that all the relevant quantities arise as norms or scalar products:
\begin{align*}
	(u,v)&\mapsto\underline{a_{u,v}}=\frac{L_G^{+/2}(\chi_u-\chi_v)}{\left\|L_G^{+/2}(\chi_u-\chi_v)\right\|},\\
	x&\mapsto\underline{y_x}=L_G^{1/2}x.
\end{align*}

Notice that both $\underline{a_{u,v}}$ and $\underline{y_x}$ are normalized (since $x^\top L_Gx$=1). Furthermore, the crucial quantity, the power of the edge $(u,v)$ arises as the square of a scalar product:
$$\langle\underline{a_{u,v}},\underline{y_x}\rangle^2=\frac{(x^\top(\chi_u-\chi_v))^2}{(\chi_u-\chi_v)^\top L_G^+(\chi_u-\chi_v)}=\frac{(x_u-x_v)^2}{R_G(u,v)}.$$

Thus we are interested in rounding $\underline{y_x}$  in a way that preserves $\langle\underline{a_{u,v}},\underline{y_x}\rangle^2$ up to small multiplicative error \emph{in all cases where it was $\approx2^{-i}$ to begin with}. Thus, it suffices to guarantee an additive error of at most $\lessapprox2^{-i}$ in our rounding scheme. This is the known problem of ``compression of approximate inner products'' and has been previously studied;~\cite{AK17} guarantees a rounding scheme whose image is of size at most $\approx\exp\left(2^i\right)$. This can be translated into a rounding scheme for $x\in\mathbb R^V$, with the same image-size, exactly as desired (see Lemma~\ref{lem:rounding}).

With the desired rounding scheme in hand, we can now use a combination of Chernoff and union bounds to prove that for all $x$ simultaneously $$x^\top L_{\wt G(F_i)}x\approx x^\top L_{G(F_i)}x.$$
Summing this over all edge-classes gives us Equation~\ref{eq:tech-overview-ordinary-main}.

For the detailed proof, which is considerably more complicated than the above sketch, see the proof of Theorem~\ref{thm:ordinary-sparsification} in Section~\ref{sec:warm-up}.

\subsection{Extension to Hypergraphs}

To adapt the previous argument to the hypergraph setting, we use the idea of balanced weight assignments from~\cite{Chen2020}. Essentially, we construct an ordinary graph $G=(V,F,\w)$ to accompany our input hypergraph $H=(V,E,w)$ by replacing each hyperedge $e$ with a clique $F_e$ over the vertices in $e$.
However, unlike in some previous works on hypergraph sparsification, the clique $F_e$ is not assigned weights uniformly, but instead the weight is carefully distributed among the edges. Intuitively, all the weight is shifted onto the most ``important'' edges. In the case of~\cite{Chen2020}, the measure of importance was ``strength'', a quantity relevant to cut sparsification, while in our case it is effective resistance.

More formally, a weighting assignment $z$ of the cliques is considered $\gamma$-balanced if for all $e\in E$
\begin{itemize}
	\item $\sum_{f\in F_e}\w(f)=w(e)$,
	\item and
	$$\gamma\cdot\min_{g\in F_e:\ \w(g)>0}R_G(g)\ge\max_{f\in F_e}R_G(f).$$
\end{itemize}
	
In words, all but the zero-weight edges of $F_e$ have approximately the same effective resistance. This allows hyperedge $e$ to inherit this effective-resistance value as its importance when sampling hyperedges. Our task is now to prove the existence of balanced weight assignments for all hypergraphs, and then to adapt the proof of Section~\ref{sec:warm-up}.

\paragraph{Finding balanced weight assignments.} In~\cite{Chen2020}, balanced weight assignments are constructed through the following intuitive process: Find a pair of edges violating the second constraint, that is $f,g\in F_e$ where $\w(g)\neq0$ and $f$ has significantly higher importance than $g$. Then shift weight from $g$ to $f$; this alleviates the constraint violation either because the importances of $g$ and $f$ become more similar, or simply because the weight of $g$ decreases to $0$. We call this resolving the imbalance of $f$ and $g$.~\cite{Chen2020} strings together such steps, carefully ordered and discretized, to eventually produce a balanced weight assignment of the input hypergraph.

However, their analysis relies heavily on a certain lemma about how ``strength'' (their measure of edge importance) behaves under weight updates. Lemma~6 of~\cite{Chen2020} states that altering the weight of an edge $f$, will not affect edges of significantly greater ``strength'' than $f$. This is not the case for effective resistances. It is easy to construct scenarios to the contrary; even ones in which altering the weight of edges of low resistance can increase the maximum effective resistance in the graph.

Thus the analysis of~\cite{Chen2020} does not extend to our setting. Instead we use a potential function argument to say that we make irreversible progress whenever we resolve the imbalance of two edges $f$ and $g$. Our choice of potential function is surprising, and is one of the main technical contributions of this paper. We define the spanning tree potential (or ST-potential) of a connected weighted ordinary graph $G=(V,F,\w)$, denoted $\Psi(G)$. For edge weights that equal $1$ uniformly (that is for unweighted graphs) it is simply the logarithm of the number of distinct spanning trees in $G$. In weighted graphs it is generalized to
$$\Psi(G)=\log\left(\sum_{T\in\mathbb T}\prod_{f\in T}\w(f)\right),$$
where $\mathbb T$ denotes the set of all spanning trees in $G$. Due to the relationship between spanning tree sampling and effective resistances (see for example~\cite{lovasz1993random}) we can prove a crucial update formula for $\Psi(G)$: if an edge $f$ has its weight changed by $\lambda\in\mathbb R$, the ST-potential increases by $\log(1+\lambda\cdot R(f))$. Since whenever we resolve the imbalance of a pair of edges, we shift weight from the edge of lower effective resistance to that of higher effective resistance, this allows us to argue that the ST-potential always increases throughout the process, which eventually terminates in a balanced weight assignment (see Algorithm~\ref{alg:greedy-balanced} and Theorem~\ref{thm:greedy-balanced}).

This proves the existence of balanced weight assignments, which suffices to show the existence of nearly linear size spectral sparsifiers for all hypergraphs. However, to improve running time (from exponential to polynomial in the input size), we introduce the novel concept of \emph{approximate} balanced weight assignments, by slightly relaxing the definition. These are still sufficient to aid in constructing spectral sparsifiers, and are faster to construct using Algorithm~\ref{alg:greedy-approx-balanced}.

For more details on the ST-potential, as well as the construction of balanced weight assignments see Section~\ref{sec:balanced}.

\paragraph{Using balanced weight assignments to construct hypergraph spectral sparsifiers.} Given a hypergraph $H=(V,E,w)$ and its balanced weight assignment $G=(V,F,\w)$ we assign importance to each hyperedge proportionally to the maximum effective resistance in $F_e$ (the clique corresponding to $e$). Thus we perform importance sampling, which samples each hyperedge independently with probability $p(e)\approx w(e)\cdot\max_{f\in F_e}R_G(f)$.

The broad strokes of the hypergraph proof in Section~\ref{sec:main} proceed very similarly to those of the proof for ordinary graphs in Section~\ref{sec:warm-up}. However, numerous details need to be figured out in order to bridge the gap between the two settings. It is interesting to note that our rounding scheme is exactly the same as in Section~\ref{sec:warm-up}, to the point of even being defined in terms of $G$, not $H$. (Indeed it is impossible to define such a rounding scheme directly in terms of $H$; Lemma~\ref{lem:rounding} relies heavily on the linear nature of the ordinary graph Laplacian.) Nevertheless, we manage to extend the approximation guarantee of the rounding scheme from edges to hyperedges (see Claim~\ref{claim:hyper-rounding}).

For the detailed analysis of hypergraph spectral sparsification through effective resistance-based importance sampling, see Section~\ref{sec:main}.

\subsection{Speed-Up}

Using a results of Sections~\ref{sec:balanced} and~\ref{sec:main} we can put together a polynomial time algorithm for spectral sparsification of hypergraphs. Simply run Algorithm~\ref{alg:greedy-approx-balanced} to produce an approximate balanced weight assignment, and then use importance sampling (Algorithm~\ref{alg:main}). The bottleneck of this procedure is constructing the weight assignment, which takes time $O(m\cdot \poly{n})$. (Given the weight assignment, it is trivial to implement importance sampling).

In Section~\ref{sec:speed-up} we reduce this to the nearly optimal $\wt O(mr + \poly{n})$. (Note that $O(mr)$ is the size of the input.) Our first step is the common trick of using a faster sparsification algorithm, but one which produces a larger output, to preprocess the input hypergraph. We use the algorithm of~\cite{Bansal2019} which -- with small modifications -- can be made to run in the desired $\wt O(mr + \poly{n})$ time. The resulting hypergraph has only polynomially many hyperedges (in $n$); however, the aspect ratio of edge weights (that is the ratio between the largest and smallest edge weights) can naturally be exponential in $n$.

Unfortunately, Algorithm~\ref{alg:greedy-approx-balanced} scales linearly in the aspect ratio of edge weights (see Theorem~\ref{thm:greedy-approx-balanced}) and so we propose another algorithm for finding a balanced weight assignment -- one specifically designed for the setting when the input graph is polynomially sparse, but has exponential aspect ratio.

Suppose our input hypergraph, $H=(V,E,w)$ has edge weights in the range $[1,\exp(n)]$. We then divide hyperedges into weight categories such as $E_i=\{e\in E|w(e)\in[n^{10(i-1)},n^{10i})\}$. We then bisect $H$ into two hypergraphs $H_1$ and $H_2$, where $H_1$ contains all hyperedges in odd numbered categories, and $H_2$ all those in even numbered categories. This results in hyergraphs ($H_1$ and $H_2$) where hyperedges fall into extremely well-separated categories; so extremely in fact, that the weight of a hyperedge in a higher category (for example $e\in E_i\subseteq H_1$) has higher weight than all hyperedges of all lower categories \emph{combined}, that is
$$w(e)\gg\sum_{e'\in E_{<i}\cap H_1}w(e').$$

We use this property to independently find weight assignments on $H_1$ and $H_2$. Informally, we go through the categories of hyperedges, from heaviest to lightest, resolving all instances of imbalance. We never return to a category once we moved on, and we prove that no amount of changes to the weight assignment of lower categories can disrupt the balance of a higher category, due to the huge discrepency in weights. For a more detailed and formal argument see Section~\ref{sec:separated}.


%% file: 300-warmup.tex

\section{Warm-Up: Ordinary Graphs}\label{sec:warm-up}

We begin by reproving the famous theorem of Spielman and Srivastava~\cite{Spielman2011a}, which states that sampling edges of a graph with probability proportional to their effective resistance (and then reweighting appropriately) results in a spectral sparsifier with high probability.
We prove a somewhat weaker version of the theorem, where we oversample by an $O(\epsilon^{-4}\log^3 n)$ factor, as opposed to the $\epsilon^{-2}\log n$ factor in the original.
Another slight difference is that our version samples every edge independently,
instead of sampling a predetermined number of edges with replacement in~\cite{Spielman2011a}. More recent proofs of the theorem of~\cite{Spielman2011a} that use the matrix Bernstein inequality as opposed to~\cite{RudelsonV07} also use the same distribution as ours. 

\begin{theorem}[A slightly weaker version of~\cite{Spielman2011a}]\label{thm:ordinary-sparsification}
Let $G=(V,F,\w)$ be a weighted ordinary graph with $n$ vertices
and let $1/n\le\epsilon\le1/2$.
For every edge $f\in F$, let $p(f)=\min(1,\lambda\cdot\w(f)\cdot R_G(f))$
for a sufficiently large factor $\lambda = \Theta(\epsilon^{-4}\log^3 n)$.
Sample each edge $f\in F$ independently with probability $p(f)$,
and give it weight $\wt\w(f)=\w(f)/p(f)$ if sampled.
The resulting graph, $\wt G=(V,\wt F,\wt\w)$ is an $\epsilon$-spectral sparsifier of $G$ with probability at least $1-O(\log n/n)$.
\end{theorem}

The original proof of this theorem used a concentration bound for matrices~\cite{RudelsonV07}
(later simplified to use the matrix Bernstein inequality)
to prove that $x^\top L_{\wt G}x$ is close to its expectation simultaneously for all $x\in\RR^n$, as required by Definition~\ref{def:sparsifier}.
This type of argument is difficult to adapt to hypergraph sparsification,
because the extension of quadratic forms to hypergraphs is highly non-linear.
We thus present an alternative proof that uses more primitive techniques
to bypass the reliance on linear algebra.

\begin{proofof}{Theorem~\ref{thm:ordinary-sparsification}}
  By Definition~\ref{def:sparsifier}, we must prove that for every $x\in\mathbb R^V$,
	\begin{equation}\label{eq:ordinary-main}
		x^\top L_{\wt G}x \in (1 \pm \epsilon)\cdot x^\top L_G x.
	\end{equation}

	We may assume without loss of generality that $x^\top L_G x = 1$. We denote the set of vectors $x$ where this is satisfied as $S^G\subseteq\mathbb R^V$.
  Furthermore, we simplify notation by denoting $L_G$ as $L$, and $L_{\wt G}$ as $\wt L$.
        Moreover, for any subset of edges $F'\subseteq F$,
        we denote the Laplacian of the subgraph of $G$ corresponding to $F'$ by $L_{F'}$,
        and similarly for the subgraph of $\wt G$ by $\wt L_{F'}$.

	It is clear from the construction of $\wt G$ that
	$$\mathbb E\left(x^\top \wt Lx\right)=x^\top Lx.$$
	Therefore, we are in effect trying to prove the concentration of a random variable around its expectation in Equation~\eqref{eq:ordinary-main}.
        Indeed, for any specific $x$, Equation~\eqref{eq:ordinary-main} holds with high probability by Chernoff bound (Theorem~\ref{thm:chernoff}). (One can consider $x^\top Lx$ as the sum of independent random variables of the form $\wt\w(u,v)\cdot(x_u-x_v)^2$.)

	In order to prove the concentration for all $x\in S^G$ simultaneously,
        we employ a net argument,
        where we ``round'' $x$ to some vector from a finite set 
        and apply a union bound on the rounded vectors.
        However, our rounding scheme is progressive 
        and has $O(\log n)$ ``levels'' with increasingly finer resolution. 
        Each $x$ will then determine a partition of the edges into levels,
        and we will prove concentration for each rounded vector
        and each level (subset of edges),
        and then apply a union bound over all these choices.

The existence of these rounding functions is guaranteed by the following lemma, which we will prove in Section~\ref{subsec:lem:rounding}.
\begin{lemma}\label{lem:rounding}
Let $G=(V,F,\w)$ be a connected weighted graph.
Then for every $i\in\mathbb N$ there exists a rounding function
$$\varphi_i:S^G\to\mathbb R^V$$
such that for all $x\in S^G$, denoting $x^{(i)}:=\varphi_i(x)$, we have:
\begin{enumerate}
\item
  The image of $\varphi_i$ is a finite set of cardinality
  $|\varphi_i(S^G)| \leq \exp\left(800C\log n\cdot2^i/\epsilon^2\right)$,
  where $C>0$ is the absolute constant from Theorem~\ref{thm:rounding}.
\item
  For every edge $f=(u,v)\in F$ such that $\max\left((x_u-x_v)^2,(x_u^{(i)}-x_v^{(i)})^2\right)\ge2^{-i}\cdot R_G(f)$,
  $$(x_u-x_v)^2 \in \left(1 \pm \frac{\epsilon}{7}\right) \cdot (x^{(i)}_u-x_v^{(i)})^2.$$
\end{enumerate}
\end{lemma}

The second guarantee of Lemma~\ref{lem:rounding} can be expressed in terms of the Laplacian of a single edge, resulting in the following corollary.
	\begin{corollary}\label{cor:rounding}
		For a rounding function $\varphi$ satisfying the guarantees of Lemma~\ref{lem:rounding}, and an edge $f=(u,v)\in F$ such that $\max\left((x_u-x_v)^2,(x_u^{(i)}-x_v^{(i)})^2\right)\ge2^{-i}\cdot R_G(f)$,
		$$x^\top L_{\{f\}}x \in \left(1 \pm \frac\epsilon7\right)\cdot x^{(i)\top}L_{\{f\}}x^{(i)} .$$
	\end{corollary}

Let us take a sequence of the rounding functions $\varphi_i$ guaranteed by Lemma~\ref{lem:rounding} for $i=1,\ldots, I := \log_2(7n/\epsilon)\le3\log n$.
For each $x\in S^G$, it yields a sequence of rounded vectors $x^{(i)}=\varphi_i(x)$ for $i=1,\ldots,I$.
Furthermore, we use $x^{(i)}$ to define the subset of edges $F_i'\subseteq F$ by
	$$F_i' := \left\{f=(u,v)\in F\ \middle|\ \left(x_u^{(i)}-x_v^{(i)}\right)^2\ge2^{-i}\cdot R_G(f)\right\}.$$
	That is, the second guarantee of Lemma~\ref{lem:rounding} holds for $\varphi_i$ on edges in $F_i'$.
	Finally, we use $\set{F_i'}_i$ to partition $F$ as follows.
        Let the base case be $F_0=F_0' := \set{ f\in F \mid p(f)=1}$,
        where we recall that $p(f)=\min(1,\lambda\cdot\w(f)\cdot R_G(f))$.
        For each $i \in [I]$, let $F_i:=F_i'\setminus\bigcup_{j=0}^{i-1}F'_j$, and finally let $F_{I+1}=F\setminus\bigcup_{i=0}^I F_i'$.
        %
	
	Thus we have partitioned $F$ in such a way that the second guarantee of Lemma~\ref{lem:rounding} applies to edges in $F_i$, with respect to $\varphi_i$. Furthermore, $F_i$ are defined in terms of $x^{(i)}$ (and $x^{(j)}$ for $j<i$) instead of $x$, so that the number of possible sets $F_i$ is finite, and bounded thanks to the first guarantee of Lemma~\ref{lem:rounding}.
	
	We establish the following claim for later use.

\begin{claim}\label{cl:power-bound}
  For all $i \in [I]$ and $f=(u,v)\in F_i$, we have
  $$(x_u^{(i)}-x_v^{(i)})^2\le3\cdot2^{-i}\cdot R_G(f).$$
\end{claim}

\begin{proof}
The second guarantee of Lemma~\ref{lem:rounding} for $\varphi_i$ applies to $f$,
and thus $(x_u^{(i)}-x_v^{(i)})^2\le(x_u-x_v)^2\cdot(1-\epsilon/7)^{-1}$.

Consider first the case $i=1$.
By Fact~\ref{fact:eff-res} and since $x\in S^G$, we have
$(x_u-x_v)^2\le R_G(f)\cdot x^\top Lx=R_G(f)$, 
and we indeed get
$(x^{(i)}_u-x^{(i)}_v)^2\le R_G(f)\cdot(1-\epsilon/7)^{-1}\le3\cdot2^{-1}\cdot R_G(f)$.

Now consider $i>1$,
and suppose towards contradiction that
$(x^{(i)}_u-x^{(i)}_v)^2>3\cdot2^{-i}\cdot R_G(f)$.
Notice that the second guarantee of Lemma~\ref{lem:rounding} also applies to $f$ for $\varphi_{i-1}$,
and thus
$(x_u^{(i-1)}-x_v^{(i-1)})^2\ge(x^{(i)}_u-x_v^{(i)})^2\cdot(1+\epsilon/7)^{-1}\cdot(1-\epsilon/7)\ge2^{-i+1}\cdot R_G(f)$.
This implies that $f\in F'_{i-1}$,
which contradicts the assumption $f\in F_i=F_i'\setminus\bigcup_{j=0}^{i-1}F_j'$.
\end{proof}

We will consider each group of edges $F_i$ separately,
and prove that $x^{\top}\wt L_{F_i}x$ concentrates around its expectation, $x^\top L_{F_i}x$.
More precisely, we will first prove concentration for every specific $(x^{(i)}, F_i)$,
and then extend the concentration to all possibilities simultaneously via union bound.
This is well-defined because each $F_i$ depends on $x^{(1)},\ldots,x^{(i)}$ but not directly on $x$. 

\paragraph{Edges in $F_0$.}
By definition, every edge $f\in F_0$ has $p(f)=1$,
and thus $x^\top\wt L_{F_0}x$ is completely deterministic and equal to $x^\top L_{F_0}x$.

\paragraph{Edges in $F_i$ for $i\in[I]$.}
Note that $F_i$ is designed so that, by Corollary~\ref{cor:rounding},
for every edge $f\in F_i$ we have
$\left|x^\top L_{\{f\}}x-x^{(i)\top}L_{\{f\}}x^{(i)}\right|\le\epsilon/7\cdot x^{(i)\top}L_{\{f\}}x^{(i)}$,
and since $\wt L_{\{f\}}$ is a multiple of $L_{\{f\}}$ similarly have $\left|x^\top\wt L_{\{f\}}x-x^{(i)\top}\wt L_{\{f\}}x^{(i)}\right|\le\epsilon/7\cdot x^{(i)\top}\wt L_{\{f\}}x^{(i)}$.
Informally, this allows us to prove concentration only for vectors $x^{(i)}$
instead of all $x$, and thus we next aim to bound the error
$$\left|x^{(i)\top}L_{F_i}x^{(i)}-x^{(i)\top}\wt L_{F_i}x^{(i)}\right|$$
for each $i\in[I]$ with high probability.
It will then remain to bound the error introduced on the remaining edges (the ones in $F_{I+1}$).

Fix $i\in[I]$ and notice that over all possible vectors $x\in S^G$,
there are only finitely many possible values for $(x^{(i)}, F_i)$.
Therefore, we can focus on a single value of $x^{(i)}$ and $F_i$,
and then use a union bound over all settings.

Let us therefore fix also $x^{(i)}$ and $F_i$.
We will use Chernoff bounds to prove that with high probability,
over the randomness of sampling edges to $\wt G$,
\begin{equation}\label{eq:medium-i-main}
  \left|x^{(i)\top}L_{F_i}x^{(i)}-x^{(i)\top}\wt L_{F_i}x^{(i)}\right|\le\frac{\epsilon}{7I}.
\end{equation}
Indeed, note that
$$
  x^{(i)\top}\wt L_{F_i}x^{(i)}
  =
  \sum_{f=(u,v)\in F_i}\wt\w(f)\cdot(x^{(i)}_u-x^{(i)}_v)^2,
$$
where $\wt\w(f)$ are independent random variables with expectation $\mathbb E(\wt\w(f))=\w(f)$.
Therefore, we can apply the Chernoff bound from Theorem~\ref{thm:chernoff}
with $\set{Z_i}_i$ being $\wt\w(f)\cdot(x_u-x_v)^2$ for each $f=(u,v) \in F_i$,
and their sum being $Z=x^{(i)\top}\wt L_{F_i}x^{(i)}$
with $\mathbb E(Z)=x^{(i)\top}L_{F_i}x^{(i)}$.
We need to set $a$ as an upper bound on $\wt\w(f)\cdot(x^{(i)}_u-x^{(i)}_v)^2$.
Observe that $\wt\w(f)$ is maximal when $f$ is sampled,
in which case it equals $\w(f)/p(f)$ where $p(f)=\lambda\cdot\w(f)\cdot R_G(f)$,
since $f\not\in F_0$.
We thus get, using Claim~\ref{cl:power-bound},
$$
  \forall f=(u,v)\in F_i,
  \qquad
  \frac{\w(f)\cdot(x_u^{(i)}-x_v^{(i)})^2}{\lambda\cdot\w(f)\cdot R_G(f)}
  = \frac1\lambda\cdot\frac{(x_u^{(i)}-x_v^{(i)})^2}{R_G(f)}
  \le\frac{3\cdot2^{-i}}\lambda
  =: a.
$$
We let $\delta := \epsilon/(14I)$, we can bound 
$$
  x^{(i)\top}L_{F_i}x^{(i)}
  \le \left(1+\frac{\epsilon}{7}\right)\cdot x^\top L_{F_i}x
  \le \left(1+\frac{\epsilon}{7}\right)\cdot x^\top Lx=1+\frac{\epsilon}{7}
  \le 2
  =: \mu.
$$
(This is true for an arbitrary preimage $x\in\varphi_i^{-1}(x^{(i)})$.)

Finally, Theorem~\ref{thm:chernoff} implies
\begin{align*}
  \mathbb P\left(\left|x^{(i)\top}L_{F_i}x^{(i)}-x^{(i)\top}\wt L_{F_i}x^{(i)}\right|\ge\frac{\epsilon}{7I}\right)
   &\le2\exp\left(-\frac{\delta^2\mu}{3\cdot a}\right)\\
   &=2\exp\left(-\frac{\tfrac{\epsilon^2}{196I^2}\cdot2}{9\cdot 2^{-i}/ \lambda}\right)\\
   &\le2\exp\left(-\frac{\epsilon^2\cdot2^i\cdot\lambda}{10000\log^2(n)}\right)\\
   &=2\exp\left(-\frac{2000C\log n\cdot2^i}{\epsilon^2}\right),
\end{align*}
where the last step by setting $\lambda = 2\cdot10^7\cdot C\log^3 n/\epsilon^4$,
where $C>0$ is the absolute constant from Theorem~\ref{thm:rounding}.

We can now use a union bound to bound the probability that Equation~\eqref{eq:medium-i-main} holds simultaneously for all values of $(x^{(i)},F_i)$.
$F_i$ depends only on $F'_j$ for $j \in [i]$, which in turn depend on $x^{(j)}$ for the same values of $j$.
By the first guarantee of Lemma~\ref{lem:rounding}, the number of possible vectors $x^{(j)}$ is at most $\exp\left(800C\log n\cdot2^j/\epsilon^2\right)$, where $C>0$ is the absolute constant from Theorem~\ref{thm:rounding}.
Therefore, the number of possible pairs $(x^{(i)},F_i)$ is at most
$$
  \prod_{j=1}^i\exp\left(\frac{800C\log n\cdot2^j}{\epsilon^2}\right)
  =\exp\left(\sum_{j=1}^i\frac{800C\log n\cdot2^j}{\epsilon^2}\right)
  \le\exp\left(\frac{1600C\log n\cdot2^i}{\epsilon^2}\right).
$$
Finally, the probability that Equation~\eqref{eq:medium-i-main} \emph{does not} hold simultaneously for all pairs $(x^{(i)},F_i)$ is at most
$$
  \exp\left(\frac{1600C\log n\cdot2^i}{\epsilon^2}\right)\cdot2\exp\left(-\frac{2000C\log n\cdot2^i}{\epsilon^2}\right)
  = 2\exp\left(-\frac{400C\log n\cdot2^i}{\epsilon^2}\right)
  \le\frac1n.
$$

\paragraph{Edges in $F_{I+1}$.}
First we show that for any $x\in S^G$ and any edge $f=(u,v)\in F_{I+1}$ we have that $(x_u-x_v)^2\le\epsilon\cdot R_G(f)/(6n)$. Suppose for contradiction that this is not the case.
Then the second guarantee of Lemma~\ref{lem:rounding} applies and $(x^{(I)}_u-x^{(I)}_v)^2\ge(x_u-x_v)^2\cdot(1-\epsilon/7)\ge\epsilon\cdot R_G(e)/(6n)\cdot(1-\epsilon/7)\ge\epsilon\cdot R_G(e)/(7n)$.
Therefore $f\in F_I'$,
which contradicts the assumption $f\in F_{I+1}$. (Here we used that $I$ was defined to be $\log_2(7n/\epsilon)$.)

Next, we would like to bound
$\left|x^{\top}\wt L_{F_{I+1}}x-x^\top L_{F_{I+1}}x\right|$
by showing that both terms are small.
First,
	\begin{align*}
		x^\top L_{F_{I+1}}x&=\sum_{f=(u,v)\in F_{I+1}}\w(f)\cdot(x_u-x_v)^2
		&\le\sum_{f\in F_{I+1}}\w(f)\cdot\epsilon\cdot \frac{R_G(f)}{6n}
		&\le\frac{\epsilon}{6n}\cdot\sum_{f\in F}\w(f)\cdot R_G(f)
		&\le\frac{\epsilon}{6},
	\end{align*}
where the last inequality uses Fact~\ref{fact:eff-res-sum}.
Second, we start similarly,
	\begin{align*}
		x^\top\wt L_{F_{I+1}}x&=\sum_{f=(u,v)\in F_{I+1}}\wt\w(f)\cdot(x_u-x_v)^2
		&\le\sum_{f\in F_{I+1}}\wt\w(f)\cdot\epsilon\cdot \frac{R_G(f)}{6n}
		&\le\frac{\epsilon}{6n}\cdot\sum_{f\in F}\wt\w(f)\cdot R_G(f),
	\end{align*}
and ideally we would like to show that $\sum\wt\w(f)\cdot R_G(f)\le2n$.
This is not always true, but it is a random event, independent of the choice of $x$, and can be shown to hold with high probability using our Chernoff bound from Theorem~\ref{thm:chernoff}.
Indeed, $\wt\w(f)\cdot R_G(f)$ are independent random variables with maximum value when $f$ is sampled, in which case $\wt\w(f)=\w(f)/p(f)$,
and thus
\begin{align*}
  \wt\w(f)\cdot R_G(f)
  &\le \frac{\w(f)\cdot R_G(f)}{p(f)}
  = \frac{\w(f)\cdot R_G(f)}{\min(1,\lambda\cdot\w(f)\cdot R_G(f))}
  =\max\left(\w(f)\cdot R_G(f),1/\lambda\right)
  \le1
  =:a  ,
\end{align*}
where the last inequality uses Fact~\ref{fact:w-R}.
We apply Theorem~\ref{thm:chernoff} by setting $\delta:=1$ and $\mu:=n$
(which we may do by Fact~\ref{fact:eff-res-sum}), and obtain
	\begin{align*}
		\mathbb P\left(\sum_{f\in F}\wt\w(f)\cdot R_G(f)\ge2n\right)&\le2\exp\left(-\frac{n}{3}\right).
	\end{align*}
Therefore, with probability at least $1-O(1/n)$,
\begin{equation}\label{eq:I+1-main}
  \left|x^\top\wt L_{F_{I+1}}x-x^\top L_{F_{I+1}}x \right|\le \frac{\epsilon}{2} .
\end{equation}

\paragraph{Putting everything together.}

By the above derivations, Equation~\eqref{eq:medium-i-main} holds for all $i$ and all $(x^{(i)},F_i)$ simultaneously, as well as Equation~\eqref{eq:I+1-main} holds with probability at least $1-O(\log n/n)$.
Assuming henceforth that this high probability event occurs,
we shall deduce that Equation~\eqref{eq:ordinary-main} holds for all $x\in S^G$.
Indeed, by the triangle inequality and Equation~\eqref{eq:I+1-main},
\begin{align*}
  \left|x^\top\wt Lx-x^\top Lx\right|
  &\le\sum_{i=0}^{I+1}\left|x^\top\wt L_{F_i}x-x^\top L_{F_i}x\right|
  \leq 0 + \sum_{i=1}^I\left|x^\top\wt L_{F_i}x-x^\top L_{F_i}x\right|+ \frac{\epsilon}{2}.
\end{align*}
Now for each $i\in[I]$,
we can approximate terms involving $x$ by $x^{(i)}$ and vice versa,
formalized by the aforementioned fact that for every $(u,v)\in F_i$ we have
$|(x_u-x_v)^2-(x^{(i)}_u-x^{(i)}_v)^2|\le\epsilon/7\cdot(x^{(i)}_u-x^{(i)}_v)^2$ (see the second condition of Lemma~\ref{lem:rounding} and the definition of $F_i'\subseteq F_i$),
and get
\begin{align*}
  \left|x^\top\wt L_{F_i}x-x^\top L_{F_i}x\right|
  &\le\left|x^\top\wt L_{F_i}x-x^{(i)\top}\wt L_{F_i}x^{(i)}\right|+\left|x^{(i)\top}\wt L_{F_i}x^{(i)}-x^{(i)\top}L_{F_i}x^{(i)}\right|+\left|x^{(i)\top}L_{F_i}x^{(i)}-x^\top L_{F_i}x\right|\\
  &\le \frac{\epsilon}{7}\cdot x^{(i)\top}\wt L_{F_i}x^{(i)}+\left|x^{(i)\top}\wt L_{F_i}x^{(i)}-x^{(i)\top}L_{F_i}x^{(i)}\right|+\frac{\epsilon}{7}\cdot x^{(i)\top}L_{F_i}x^{(i)}\\
\intertext{now we use the triangle inequality,  }
  &\le \frac{\epsilon}{7}\cdot x^{(i)\top}L_{F_i}x^{(i)} + \left(1+\frac{\epsilon}{7}\right)\cdot\left|x^{(i)\top}\wt L_{F_i}x^{(i)}-x^{(i)\top}L_{F_i}x^{(i)}\right|+\frac{\epsilon}{7}\cdot x^{(i)\top}L_{F_i}x^{(i)}\\
\intertext{and now we crucially use Equation~\eqref{eq:medium-i-main}, }
  &\le\left(1+\frac{\epsilon}{7}\right)\cdot\frac{\epsilon}{7I}+\frac{2\epsilon}{7}\cdot \left(1-\frac{\epsilon}{7}\right)^{-1}\cdot x^\top L_{F_i}x\\
  &\le\frac{\epsilon}{6I}+\frac{2\epsilon}{6}\cdot x^\top L_{F_i}x.
\end{align*}
Substituting this into our previous bound, we obtain
\begin{align*}
  \left|x^\top\wt Lx-x^\top\wt Lx\right|
  &\le\sum_{i=1}^I\left(\frac{\epsilon}{6I}+\frac{\epsilon}{3}\cdot x^\top L_{F_i}x\right) + \frac{\epsilon}{2}
  \le\frac{\epsilon}{6}+\frac{\epsilon}{3}\cdot x^\top Lx+\frac{\epsilon}{2}
  =\epsilon\cdot x^\top Lx,
\end{align*}
where the last equality uses $x^\top Lx = 1$.
This completes the proof of Theorem~\ref{thm:ordinary-sparsification}.
\end{proofof}

\subsection{Proof of Lemma~\ref{lem:rounding}}\label{subsec:lem:rounding}

To prove Lemma~\ref{lem:rounding}, we use the following Theorem:

\begin{theorem}[Theorem VI.1 of~\cite{AK17}]\label{thm:rounding}
	Let $a_1,\ldots,a_m \in \mathbb R^n$ be vectors of norm at most $1$ and let $\eta\in(0,1)$.
	Then, over all vectors $y \in \mathbb R^n$ of norm at most $1$,
        the number of possible values of the ``rounded vector'' 
	$$\left(\left\lfloor\frac{\langle a_1,y\rangle}{\eta}\right\rfloor,\left\lfloor\frac{\langle a_2,y\rangle}{\eta}\right\rfloor,\ldots,\left\lfloor\frac{\langle a_k,y\rangle}{\eta}\right\rfloor\right)$$
	is at most
	$\exp\left(\frac{C\log m}{\eta^2}\right)$
	for some absolute constant $C>0$.
\end{theorem}

\begin{remark}
	In fact, the original theorem (Theorem 6.1 in~\cite{AK17}) is stated with stronger requirements on $m$ and $\eta$, and a stronger consequence. However, we can easily get the weaker upper bound of $\exp(O(\log m/\eta^2))$ stated in Theorem~\ref{thm:rounding} of this paper, by setting the variables appropriately: $\varepsilon:=\eta$, $n:=\max(m,1/\eta^2)$, and $k:=n$, where the left hand side always represents their variable names and the right hand side ours.
\end{remark}

\begin{proofof}{Lemma~\ref{lem:rounding}}
We use the idea of resistive embedding introduced in~\cite{}.
Note that $L=L_G$ is a positive semidefinite matrix, and we denote by $L^{+/2}$  the square root of its Moore-Penrose pseudo-inverse.
For each (unordered) vertex pair $(u,v)$, let $b_{u,v}\in\mathbb R^V$ be the vector with all zero coordinates, except for the coordinates associated with $u$ and $v$, which are $1$ and $-1$ (ordered arbitrarily).
With each vertex pair $(u,v)$, we associate the vector
$$a_{u,v}=\frac{L^{+/2}b_{u,v}}{\|L^{+/2}b_{u,v}\|_2}.$$
Furthermore, we associate with each $x\in S^G$ the vector $y_x=L^{1/2}x$.

We can then apply Theorem~\ref{thm:rounding} to $\set{a_{u,v} \mid (u,v)\in\binom{V}{2} }$ and all possible $y_x$,
setting $\eta=\epsilon\cdot2^{-i/2}/20$.
Indeed, $a_{u,v}$ is normalized by definition,
and also $y_x$ is normalized because $x\in S^G$ and thus
$$\|y_x\|_2^2=x^\top L^{1/2}L^{1/2}x=x^\top Lx=1. $$
For each possible value of the rounded vector
$$\left(\left\lfloor\frac{\langle a_{u,v},y_x\rangle}{\eta}\right\rfloor\right)_{(u,v)\in\binom{V}{2}}$$
choose a representative $x\in S^G$, 
and let $\varphi_i$ map each $x\in S^G$ to its representative
(i.e., with the same rounded vector).
Then by Theorem~\ref{thm:rounding}, the image of $\varphi_i$ is of size
$|\varphi_i(S^G)| \leq \exp\left(800C\log n\cdot2^i/\epsilon^2\right)$,
as claimed.
Recall that we denote $\varphi_i(x)$ by $x^{(i)}$; then
$$\left(\left\lfloor\frac{\langle a_{u,v},y_x\rangle}{\eta}\right\rfloor\right)_{(u,v)\in\binom{V}{2}}=\left(\left\lfloor\frac{\langle a_{u,v},y_{x^{(i)}}\rangle}{\eta}\right\rfloor\right)_{(u,v)\in\binom{V}{2}}. $$
It follows that for all $f=(u,v)\in F$ and all $x\in S^G$,
\begin{equation}\label{eq:rounding-add-error}          
  \left|\langle a_{u,v},y_x\rangle-\langle a_{u,v},y_{x^{(i)}}\rangle\right|\le\eta .
\end{equation}
Furthermore, $b_{u,v}$ is perpendicular to the null-space of $L$
(which is spanned by the all-ones vector because $G$ is connected),
thus $L^{1/2}L^{+/2}b_{u,v}=b_{u,v}$ 
and
\begin{equation}  \label{eq:auvyx}
  \langle a_{u,v},y_x\rangle^2
  = \frac{\left(x^{\top}L^{1/2}L^{+/2}b_{u,v}\right)^2}{b_{u,v}^\top L^{+/2}L^{+/2}b_{u,v}}
  = \frac{\left(x^{\top}b_{u,v}\right)^2}{b_{u,v}^\top L^+b_{u,v}}
  = \frac{(x_u-x_v)^2}{R_G(u,v)}.
\end{equation}

To prove the second guarantee of Lemma~\ref{lem:rounding},
let $f=(u,v)\in F$ and $x\in S^G$ and consider first the case
$\left(x_u^{(i)}-x_v^{(i)}\right)^2\ge2^{-i}\cdot R_G(f)$,
which by~\eqref{eq:auvyx} is equivalent to $\langle a_{u,v},y_{x^{(i)}}\rangle^2\ge2^{-i}$. 
This means that the absolute error bound $\eta$ in Equation~\eqref{eq:rounding-add-error} implies a relative error bound, namely, 
$$
  \left|\langle a_{u,v},y_x\rangle - \langle a_{u,v},y_{x^{(i)}}\rangle\right|
  \le \eta
  = \frac{\epsilon\cdot2^{-i/2}}{20}
  \le \frac{\epsilon}{20}\cdot \left|\langle a_{u,v},y_{x^{(i)}}\rangle\right| .
$$
The other case $\left(x_u-x_v\right)^2 \ge 2^{-i}\cdot R_G(f)$
is similar up to constants;
by~\eqref{eq:auvyx}, this case is equivalent to $\langle a_{u,v},y_x\rangle^2\ge2^{-i}$,
and thus 
\begin{align*}
  \left|\langle a_{u,v},y_x\rangle - \langle a_{u,v},y_{x^{(i)}}\rangle\right|
  &\le \eta
  = \frac{\epsilon\cdot2^{-i/2}}{20}
  \le \frac{\epsilon}{20}\cdot \left|\langle a_{u,v},y_x\rangle\right| ,
  \\
\intertext{which implies}
  \left|\langle a_{u,v},y_x\rangle - \langle a_{u,v},y_{x^{(i)}}\rangle\right|
  &\le \frac{\epsilon}{20}\cdot \left(1-\frac{\epsilon}{20}\right)^{-1} \left|\langle a_{u,v},y_x\rangle\right| 
  \le \frac{\epsilon}{16}\cdot \left|\langle a_{u,v},y_{x^{(i)}}\rangle\right| .
\end{align*}

Now in both cases, 
\begin{align*}
  \left|\langle a_{u,v},y_x\rangle^2 - \langle a_{u,v},y_{x^{(i)}}\rangle^2\right|
  &=\left|\langle a_{u,v},y_x\rangle-\langle a_{u,v},y_{x^{(i)}}\rangle\right|\cdot\left|\langle a_{u,v},y_x\rangle+\langle a_{u,v},y_{x^{(i)}}\rangle\right|
  \\
  &\le\frac{\epsilon}{16} \cdot |\langle a_{u,v},y_{x^{(i)}}\rangle| \cdot\left(2+\frac{\epsilon}{16}\right) \cdot |\langle a_{u,v},y_{x^{(i)}}\rangle| \\
  &\le\frac{\epsilon}{7}\cdot\langle a_{u,v},y_{x^{(i)}}\rangle^2. 
\end{align*}
Using~\eqref{eq:auvyx} and scaling by $R_G(u,v)$, we can write this as
$(x_u-x_v)^2\in (1 \pm \epsilon/7)\cdot (x^{(i)}_u-x_v^{(i)})^2$,
which completes the proof of Lemma~\ref{lem:rounding}. 
\end{proofof}


%% file: 400-balanced.tex

\section{$\gamma$-Balanced Weight Assignments}\label{sec:balanced}

Our strategy for generalizing the techniques of Section~\ref{sec:warm-up} to hypergraphs is similar to that of~\cite{Chen2020}.
Intuitively, we wish to replace each hyperedge of the input hypergraph with a weighted clique in such a way that the ``importance'' of each edge in the same clique is roughly the same.
However, our measure of importance is effective resistance, whereas in~\cite{Chen2020} it is the \emph{strength} of the edge (see~\cite{Benczur2015}), which is a measure particularly useful to cut sparsification.
Specifically, we use the following definition.
\begin{definition}\label{def:balanced}
	Given a hypergraph $H=(V,E,w)$, a weight assignment of $H$ is a weighted (ordinary) graph $G=(V,F,\w)$ such that
	\begin{itemize}
	\item $F$ is the multiset $\bigcup_{e \in E} F_e$, where $F_e$ is a set of edges forming a clique on the support of $e$\;
	\item $\sum_{f\in F_e}\w(f)=w(e)$.
	\end{itemize}
	Note that this definition allows for parallel edges in $G$.

	Moreover, if $G$ satisfies
	$$\gamma\cdot\min_{g\in F_e:\ \w(g)>0}R_G(g)\ge\max_{f\in F_e}R_G(f)$$
	for $\gamma > 1$, then we call it $\gamma$-balanced.
\end{definition}
The goal of this section is to show the existence of constant-balanced weight assignments for all weighted hypergraphs, and to define an efficient algorithm that outputs such a weight assignment.

\subsection{Spanning Tree Potential}

In order to show the existence (and give an efficient construction) of $\gamma$-balanced weight assignments we introduce the concept of spanning tree potentials for weighted ordinary graphs.
\begin{definition}[Spanning tree potential, ST-potential for short]
	For a connected weighted graph $G=(V,F,\w)$ let $\mathbb{T}(G)$ be the set of all spanning trees of $G$.
	Then we define the spanning tree potential of $G$ as
	\begin{equation}\label{eq:Psi-formula}
		\Psi(G)=\log\left[\sum_{T\in\mathbb{T}(G)}\prod_{f\in T}\w(f)\right].
	\end{equation}
\end{definition}
\begin{remark}
  Note that the value of ST-potential stays the same after replacing parallel edges $f_1$ and $f_2$ with a single edge $f$ of weight $\w(f) := \w(f_1) + \w(f_2)$.
\end{remark}

We formalize the concept of edge updates to graphs and see how those updates affect ST-potential.
\begin{definition}
	If $G=(V,F,\w)$ is a weighted graph, $\lambda\in\mathbb R$, and $f\in\binom{V}{2}$, then $G+\lambda\cdot f$ is the weighted graph $(V,F',\w')$,  where the weight of $f$ is increased by $\lambda$. Formally,
	\begin{itemize}
		\item $F'=F$, $\w'(f)=\w(f)+\lambda$, and $\w'(g)=\w(g)$ for all $g\neq f$, if $f\in F$,
		\item or $F'=F+\{f\}$, $\w'(f)=\lambda$, and $\w'(g)=\w(g)$ for all $g\neq f$, if $f\not\in F$.
	\end{itemize}
	Note that the definition applies to $\lambda<0$ as well, but in this case $f$ must be present in the graph with weight at least $|\lambda|$ in order for $G+\lambda\cdot f$ to be valid.
\end{definition}

\begin{lemma}\label{lem:Psi-update}
  For any weighted graph $G$, $\lambda\in\mathbb R$ and $f\in\binom{V}{2}$ such that $G+\lambda\cdot f$ is well defined, we have
		\begin{equation}\label{eq:Psi-main}\Psi(G+\lambda\cdot f)=\Psi(G)+\log\left(\lambda R_G(f) + 1\right),\end{equation}
\end{lemma}
\begin{proof}
	We use the famous result that if we sample a spanning tree $\mathcal T$ randomly from $\mathbb{T}(G)$ such that
	$$\mathbb P(\mathcal T = T)\propto\prod_{g\in T}\w(g),$$
	then the marginal probability $\mathbb P(f\in T)$ is $\w(f)\cdot R(f)$ for all $f\in F$ (see, e.g.,~\cite{lovasz1993random}).
	Let $\mathcal T$ be a random variable drawn from such a distribution. By the definition of the distribution of $\mathcal T$ we have that
	\begin{equation}\label{eq:mathcal-t}
	\mathbb P(f\in\mathcal T)=\sum_{T \in \mathbb{T}(G):T\ni f}\mathbb P(\mathcal T = T)=\frac{\sum_{T \in \mathbb{T}(G):T\ni f}\prod_{g\in T}\w(g)}{\sum_{T\in\mathbb{T}(G)}\prod_{g\in T}\w(g)}
	\end{equation}
	Therefore, we have for the weight function $\w'$ of $G + \lambda \cdot f$,
	\begin{align*}
		\Psi(G + \lambda \cdot f)
		& = \log\left[\sum_{T\in\mathbb{T}(G')}\prod_{g\in T}\w'(g)\right] \\
		&=\log\left[\sum_{T\in\mathbb{T}(G)}\prod_{g\in T}\w(g)\cdot \left(1+\frac{\lambda}{\w(f)}\right)^{\mathbbm1(g=f)}\right] \\
		&=\log\left[\sum_{T\in\mathbb T}\left(1+\frac{\lambda}{\w(f)}\right)^{\mathbbm1(f\in T)}\prod_{g\in T}\w(g)\right] \\
		&=\log\left[\sum_{T\in\mathbb{T}(G)}\prod_{g\in T}\w(g)+\frac{\lambda}{\w(f)}\cdot\sum_{T \in \mathbb{T}(G):T\ni f}\prod_{g\in T}\w(g)\right]. \\
		\intertext{We can transform the second term by Equation~\eqref{eq:mathcal-t} to continue:}
		&=\log\left[\sum_{T\in\mathbb{T}(G)}\prod_{g\in T}\w(g)+\frac{\lambda}{\w(f)}\cdot\mathbb P(f\in\mathcal T)\cdot\sum_{T\in\mathbb{T}(G)}\prod_{g\in T}\w(g)\right] \\
		&=\log\left[\sum_{T\in\mathbb{T}(G)}\prod_{g\in T}\w(g)\right] + \log\left(\frac{\lambda}{\w(f)}\cdot\mathbb P(f\in\mathcal T) + 1\right) \\
		&=\Psi(G) + \log(\lambda R_G(f) + 1)
	\end{align*}
	as claimed. Note that the above calculation is legitimate for positive $\lambda$ as well as negative.
\end{proof}

\subsection{Existence of $\gamma$-Balanced Weight Assignments}

We will now use ST-potential to construct a $\gamma$-balanced weight assignment of an arbitrary hypergraph. We can analyze a simple greedy algorithm, which identifies edge pairs contradicting the $\gamma$-balancedness condition in Definition~\ref{def:balanced}, and simply shifts weight from the one with smaller effective resistance to the one with the larger effective resistance (see Algorithm~\ref{alg:greedy-balanced}).
One can show that such an step can be designed to only increase the ST-potential of a weight assignment, and thus the algorithm eventually terminates, returning a $\gamma$-balanced weight assignment.

\begin{algorithm}[H]
	\begin{algorithmic}[1]
		\caption{Algorithm for constructing a $\gamma$-balanced weight assignment.}\label{alg:greedy-balanced}
		\Procedure{GreedyBalancing}{$H=(V,E,w),\gamma$}
		\State For all $e\in E$ and for all $f\in F_e$, initialize $\w(f)$ to $w(e)/\binom{|e|}{2}$
		\State $G\gets(V,\bigcup_{e \in E} F_e,\w)$
		\While{$G$ is not a $\gamma$-balanced weight assignment of $H$}
		\State Select $e\in E$, and $f,g\in F_e$, such that $R_G(f)>\gamma\cdot R_G(g)$ and $\w(g)>0$
		\State $\lambda\gets\min\left(\w(g),(\gamma - 1)/(2\gamma\cdot R_G(g))\right)$\label{line:select-problem-edge}
		\State $\w(f)\gets \w(f)+\lambda$
		\State $\w(g)\gets \w(g)-\lambda$
		\EndWhile
		\State \Return $G$
		\EndProcedure
	\end{algorithmic}
\end{algorithm}

The following tells how much the effective resistance of an edge changes by updating the weight of another edge.
Although the proof is simple and this result is already known, we include the proof for completeness.
\begin{lemma}\label{lem:update-formula}
	If $G=(V,F,\w)$ is a weighted graph, let $\lambda\in\mathbb R$, and $f\in\binom{V}{2}$, then for any $g\in\binom{V}{2}$
	$$R_{G+\lambda\cdot f}(g)=R_G(g)-\frac{\lambda\cdot\left(b_g^\top L_G^+b_f\right)^2}{1+\lambda\cdot R_G(f)}.$$
\end{lemma}
\begin{proof}
	Note that
	$$L_{G+\lambda\cdot f}=L_G+\lambda b_f b_f^\top.$$
	Therefore, by the Sherman-Morrison formula for Moore-Penrose pseudoinverse (see for example~\cite{Meyer73}), we can expand the formula for the effective resistance of $g$:
	\begin{align*}
		R_{G+\lambda\cdot f}(g)&=b_g^\top L_{G+\lambda\cdot f}^+b_g
		=b_g^\top{\left(L_G+\lambda b_f b_f^\top\right)}^+b_g
		=b_g^\top\left(L_G^+-\frac{L_G^+b_f\lambda b_f^\top L_G^+}{1+\lambda b_f^\top L_G^+b_f}\right)b_g\\
		&=b_g^\top L_G^+b_g-\frac{\lambda b_g^\top L_G^+b_f b_f^\top L_G^+b_g}{1+\lambda b_f^\top L_G^+b_f}
		=R_G(g)-\frac{\lambda\cdot\left(b_g^\top L_G^+b_f\right)^2}{1+\lambda\cdot R_G(f)}.
		\qedhere
	\end{align*}
\end{proof}

\begin{lemma}\label{lem:greedy-step}
	Let $G=(V,F,\w)$ be a weighted graph and let $\gamma>1$. Let $f,g$ be two edges in $F$ such that $R_G(f)>\gamma\cdot R_G(g)$. Then for any $\lambda\le\w(g)$, shifting $\lambda$ weight from $g$ to $f$ results in an increase of at least
	$$\log\left(1+\lambda\gamma\cdot R_G(g)-\lambda\cdot R_G(g)-\lambda^2\gamma\cdot R_G(g)^2\right)$$
	in the ST-potential of $G$.
\end{lemma}

\begin{proof}

We simply apply the update formula for ST-potential (Lemma~\ref{lem:Psi-update}) twice, along with Lemma~\ref{lem:update-formula}. For simplicity, we use $t_{fg}$ to denote $b_f^\top L_G^+b_g$.
Then the increase in ST-potential is
	\begin{align*}&\log\left(\lambda\cdot R_G(f)+1\right)+\log\left(-\lambda\cdot R_{G+\lambda\cdot f}(g)+1\right)\\
=&\log\left(\lambda\cdot R_G(f)+1\right)+\log\left(-\lambda\cdot\left(R_G(g)-\frac{\lambda t_{fg}^2}{1+\lambda R_G(f)}\right)+1\right)\\
=&\log\left[\left(\lambda\cdot R_G(f)+1\right)\cdot\left(-\lambda\cdot\left(R_G(g)-\frac{\lambda t_{fg}^2}{1+\lambda R_G(f)}\right) + 1\right)\right]\\
=&\log\left(1+\lambda\cdot R_G(f)-\lambda\cdot R_G(g)-\lambda^2\cdot R_G(f)R_G(g)+\lambda^2t_{fg}^2\right)\\
\ge&\log\left(1+\lambda\gamma\cdot R_G(g)-\lambda\cdot R_G(g)-\lambda^2\gamma\cdot R_G(g)^2\right),
\end{align*}
as claimed.
\end{proof}

\begin{theorem}\label{thm:greedy-balanced}
	For $\gamma>1$, Algorithm~\ref{alg:greedy-balanced} terminates and returns a $\gamma$-balanced weight assignment of $H$.
\end{theorem}

\begin{proof}
	It is clear by the condition of the while-loop that if Algorithm~\ref{alg:greedy-balanced} terminates, it returns a $\gamma$-balanced weight assignment.
	Also, the condition $\sum_{f\in F_e}\w(f)=w(e)$ is never violated.
	Therefore, in Line~\ref{line:select-problem-edge}, there must indeed always be some $e\in E$ and some $f,g\in F_e$ such that $R_G(f)>\gamma\cdot R_G(g)$ and $\w(g)>0$, otherwise $G$ would already be $\gamma$-balanced.
	It remains to prove that Algorithm~\ref{alg:greedy-balanced} always terminates.

	Let us examine the evolution of $\Psi(G)$ throughout the algorithm. First, note that $G$ never becomes disconnected, and hence $\Psi(G)$ always remains defined. Indeed, in order for $G$ to become disconnected, we would have to set $\lambda$ to $\w(g)$ for a \emph{bridge} $g$. However, if $g$ is a bridge, $R_G(g)=1/\w(g)$ by Fact~\ref{fact:w-R}, and $\lambda$ is set instead to $(\gamma-1)\w(g)/(2\gamma)<\w(g)$.
	
	In each iteration of the while-loop (which we call a step) we move $\lambda$ weight from some edge $g$ to some other edge $f$. By Lemma~\ref{lem:greedy-step} this results in a change of
	$$\log\left(1+\lambda\gamma\cdot R_G(g)-\lambda\cdot R_G(g)-\lambda^2\gamma\cdot R_G(g)^2\right).$$

	Here we distinguish between the following two cases: If $\lambda=(\gamma-1)/(2\gamma\cdot R_G(g))$ exactly, then the increase in $\Psi(G)$ is at least $\log(1+(\gamma-1)^2/(2\gamma) - (\gamma-1)^2/(4\gamma))\ge\log(1+(\gamma-1)^2/(4\gamma))=:c_\gamma$. On the other hand, if $\lambda=\w(g)\le(\gamma-1)/(2\gamma\cdot R_G(g))$ then the increase in $\Psi(G)$ is at least $\log(1+\lambda\gamma\cdot R_G(g)-\lambda\cdot R_G(g)-\lambda(\gamma-1)\cdot R_G(g)/2)>\log(1)=0$.

	Overall there are two possibilities each step:
	\begin{enumerate}
		\item $\Psi(G)$ increases by at least $c_\gamma>0$,
		\item or $\w(g)$ becomes $0$, and $\Psi(G)$ increases by a positive amount.
	\end{enumerate}

	Let the initial setting of $G$ (before the while-loop) be $G_0$. Let $G_\infty$ be the complete graph on $V$ with uniform edge weights of $\sum_{e\in E}w(e)$ on each edge. Since $G$ always satisfies
	$$\sum_{f\in F}\w(f)=\sum_{e\in E}\sum_{f\in F_e}\w(f)=\sum_{e\in E}w(e),$$
	$\Psi(G)$ will always be less than $\Psi(G_\infty)$ by monotonicity of $\Psi$.
	Thus, there can be at most $(\Psi(G_\infty)-\Psi(G_0))/c_\gamma$ steps of type 1. Therefore, after a certain point, there can only be steps of type 2; we focus on this stage of the algorithm.

	We further categorize steps of type 2 based on the initial weight of $f$:

	\begin{UR}
		\item The initial weight of $f$ is greater than $0$,
		\item the initial weight of $f$ is exactly $0$.
	\end{UR}

	 Steps of type 2a increase the total number of edges of weight exactly $0$, since by definition, neither $f$ nor $g$ starts out with weight $0$, but $\w(g)$ becomes $0$. On the other hand, steps of type 2b do not decrease the total number of edges of weight $0$. Therefore, after a certain point, there can only be steps of type 2b; we focus on this stage of the algorithm.

	At this point, the set of all edge weight values, that is
	$$\bigcup_{f\in F}\{\w(f)\},$$
	remains unchanged. Indeed at every step we simply switch the values of $\w(f)$ and $\w(g)$. Therefore, there are only a finite number of possible states for $G$ to be in. None of these can be repeated, as $\Psi(G)$ increases by a positive amount after each step, and therefore, Algorithm~\ref{alg:greedy-balanced} must terminate, returning a $\gamma$-balanced weight assignment of $H$.
\end{proof}

This proves the existence of $\gamma$-balanced weight assignment, which suffices to show the existence of nearly-linear-sized spectral sparsifiers (as we will see in Section~\ref{sec:main}). Unfortunately, the above proof shows no bound on the running time of Algorithm~\ref{alg:greedy-balanced} beyond $2^{O(m)}$.

\subsection{Polynomial-Time Construction}

In this section, we introduce a relaxation of the concept of $\gamma$-balanced weight assignment. This will still be sufficient to get spectral sparsifiers of nearly linear size, while also allowing the greedy algorithm to terminate in a polynomial number of steps.

\begin{definition}\label{def:weighted-approx-balanced}
	For $0<\eta \le1$ and $\gamma > 1$, an $\eta$-approximate $\gamma$-balanced weight assignment of $H=(V,E,w)$ is defined exactly as above in Definition~\ref{def:balanced}, except with the final condition relaxed to
	$$\gamma\cdot\min_{g\in F_e:\ \w(g)\geq \eta\cdot w(e)}R_G(g)\ge\max_{f\in F_e}R_G(f).$$
	That is, we allow edges not only of weight $0$, but also of weight less than $\eta \cdot w(e)$ to be small outliers in terms of effective resistance.
\end{definition}

We modify Algorithm~\ref{alg:greedy-balanced} to search for approximate $\gamma$-balanced weight assignments.

\begin{algorithm}[H]
	\begin{algorithmic}[1]
		\caption{Algorithm for constructing an $\eta$-approximate $\gamma$-balanced weight assignment.}\label{alg:greedy-approx-balanced}
		\Procedure{GreedyApproxBalancing}{$H=(V,E,w),\gamma,\eta$}
		\State For all $e\in E$ and for all $f\in F_e$, initialize $\w(f)$ to $w(e)/\binom{|e|}{2}$
		\State $G\gets(V,\bigcup F_e,\w)$
		\While{$G$ is not an $\eta$-approximate $\gamma$-balanced weight assignment of $H$}
		\State Select $e\in E$, and $f,g\in F_e$, such that $R_G(f)>\gamma\cdot R_G(g)$ and $\w(g)>\eta\cdot w(e)$\label{line:greedy-approx:select-problem-edge}
		\State $\lambda\gets\min\left(\w(g),1/(2R_G(g))\right)$
		\State $\w(f)\gets\w(f)+\lambda$
		\State $\w(g)\gets\w(g)-\lambda$
		\EndWhile
		\State \Return $G$
		\EndProcedure
	\end{algorithmic}
\end{algorithm}

\begin{theorem}\label{thm:greedy-approx-balanced}
	Let $H=(V,E,w)$ be a weighted hypergraph Then for $\gamma\ge4$ and $\eta > 0$, Algorithm~\ref{alg:greedy-approx-balanced} terminates \emph{within $M/(\eta w_{\min})\cdot\poly{n\log(M/w_{\min})}$ rounds} and returns an $\eta$-approximate $\gamma$-balanced weight assignment for $H$, where $w_{\min} := \min_{e\in E}w(e)$ and $M :=\sum_{e\in E}w(e)$.
\end{theorem}

\begin{proof}
	It is clear by the condition of the while-loop that if Algorithm~\ref{alg:greedy-approx-balanced} terminates, then it returns an $\eta$-approximate $\gamma$-balanced weight assignment.
	Also, the condition $\sum_{f\in F_e}\w(f)=w(e)$ is never violated.
	Therefore, in Line~\ref{line:greedy-approx:select-problem-edge}, there must indeed always be some $e\in E$ and some $f,g\in F_e$ such that $R_G(f)>\gamma$ and $\w(g)>\eta\cdot w(e)$, otherwise $G$ would already be $\eta$-approximately $\gamma$-balanced.
	It remains to prove that Algorithm~\ref{alg:greedy-approx-balanced} terminates within $M/(\eta w_{\min})\cdot\poly{n\log(M/w_{\min})}$ rounds.

	Let $G_0$ be the starting graph of Algorithm~\ref{alg:greedy-approx-balanced}, and $G_\infty$ be the complete graph on $V$ with uniform edge weights of $M$. Since $\Psi(G_\infty)$ is always greater than $\Psi(G)$ at every moment in the algorithm, we can upper bound the total increase in ST-potential throughout the algorithm by $\Psi(G_\infty)-\Psi(G_0)$. Let us upper bound $\Psi(G_\infty)-\Psi(G_0)$. To this end we will produce a sequence of updates to $G_0$ resulting in $G_\infty$, and we will upper bound the contribution of each update.

	Since $\min_{e\in E}w(e)=w_{\min}$, $\min_{f\in F}\w(f)\ge w_{\min}/n^2$ in $G_0$. Furthermore, since $H$ is connected, $G_0$ is connected as well and must contain a spanning tree with edges of weight at least $w_{\min}/n^2$ and therefore $\min_{u,v\in V}R_{G_0}(u,v)\le n/(w_{\min}/n^2) = n^3/w_{\min}$.
	To transform $G_0$ to $G_\infty$, we simply add to each vertex pair $(u,v)$ a sufficient amount of weight to make $w(u,v)=M$. This is an update of $+\lambda\cdot(u,v)$ with $\lambda\le M$ and $R(u,v)\le n^3/w_{\min}$, which contributes by at most $\log(Mn^3/w_{\min}+1)$ to $\Psi$ by Lemma~\ref{lem:Psi-update}.
	The total contribution of all such updates on the way from $G_0$ to $G_\infty$ is at most
	$$\sum_{u,v\in V}\log\left(\frac{Mn^3}{w_{\min}} + 1\right)=\poly{n\log\left(\frac{M}{w_{\min}}\right)}.$$

	We now lower bound the minimum increase of $\Psi$ after each step of Algorithm~\ref{alg:greedy-approx-balanced}. We are able to do this thanks to the modification in Definition~\ref{def:weighted-approx-balanced} of \emph{approximate} $\gamma$-balanced weight assignment.

	Due to Lemma~\ref{lem:greedy-step}, the contribution of each update to the potential $\Psi(G)$ is at least
	$$\log\left(1+\lambda\gamma\cdot R_G(g)-\lambda\cdot R_G(g)-\lambda^2\gamma\cdot R_G(g)^2\right)$$
	At each step, $\lambda$ is set to either $\w(g)$ or $1/(2R_G(g))$.
	If $\lambda = 1/(2R_G(g))$, the increase in $\Psi(G)$ is at least $\log(1+(\gamma-1)/2 - \gamma/4)\ge\log(5/4)$. On the other hand, if $\lambda = \w(g)\le1/(2R_G(g))$, then the increase in $\Psi(G)$ is at least
	\begin{align*}
		\log(1+\lambda\gamma\cdot R_G(g)-\lambda\cdot R_G(g)-\lambda\gamma\cdot R_G(g)/2)&\ge\log(1+\lambda\cdot R_G(g))\\
		&=\log(1+\w(g)\cdot R_G(g))\\
		&\ge\log(1+\eta\cdot w_{\min}\cdot R_G(g)).
	\end{align*}
	To lower bound this, note that for all $g$
	$$R_G(g)\ge R_{G_\infty}(g)=\frac2{Mn}.$$
	This gives us that after each round of the algorithm $\Psi(G)$ increases by at least
	\[
		\log(1+\eta\cdot w_{\min}\cdot R_G(g)) \geq \frac{2 \cdot 2\eta w_{\min}/(nM)}{2 + 	2\eta w_{\min}/(nM)} \geq \frac{4\eta w_{\min}/(nM)}{4} = \frac{\eta w_{\min}}{nM},
	\]
	where we used $\log(1+x) \geq (2x)/(2+x)$ for $x \geq 0$ in the first inequality and we assumed $n$ is sufficiently large in the second inequality.
	Thus the algorithm takes at most $M/(\eta w_{\min})\cdot\poly{n\log(M/w_{\min})}$ steps.
\end{proof}


%% file: 500-mainresult.tex

\section{Hypergraph Sparsification}\label{sec:main}
In this section, we prove the existence of a spectral sparsifier with a nearly linear number of hyperedges, that is, the first part of Theorem~\ref{thm:main}.
We discuss efficient construction of spectral sparsifier in Section~\ref{sec:speed-up}.

To construct a spectral sparsifier, we first produce a $1/n^2$-approximate $\gamma$-balanced weight assignment for the input hypergraph, where $\gamma \geq 4$.
We can use Algorithm~\ref{alg:greedy-approx-balanced} for this. We then assign to each hyperedge $e$ importance equal to the maximum effective resistance in $F_e$ (see Definition~\ref{def:balanced}). We then perform typical importance sampling on the hyperedges, oversampling by a factor of $\lambda$ (to be set later).
A formal description is given in Algorithm~\ref{alg:main}.

\begin{algorithm}[H]
  \caption{$\epsilon$-spectral sparsification for a hypergraph, using importance sampling.}\label{alg:main}
	\begin{algorithmic}[1]
		\Procedure{Sparsification}{$H=(V,E,w),G=(V,F,\w),\epsilon,\lambda$}
		\State $\wt H=(V,\wt E,\wt w)\gets(V,\emptyset,0)$
		\ForAll{$e\in E$}
			\State $R^{\max}(e)\gets\max_{f\in F_e}R_G(f)$
			\State $p(e)\gets\min(1,w(e)\cdot R^{\max}(e)\cdot\lambda)$
			\State With probability $p(e)$, $\wt E\gets E\cup\{e\}$ and $\wt w(e)\gets w(e)/p(e)$
		\EndFor
		\State \Return $\wt H$
		\EndProcedure
	\end{algorithmic}
\end{algorithm}

In the rest of this section, we show the correctness of this approach. We first bound the size of the hypergraph output by Algorithm~\ref{alg:main}. We then prove that the output $\wt H$ is indeed an $\epsilon$-spectral sparsifier of $H$---this is the technical core of the section.

\begin{lemma}\label{lem:hypergraph_size}
Let $H=(V,E, w)$ be a weighted hypergraph and let $G=(V,F,\w)$ be its $1/n^2$-approximate $\gamma$-balanced weight assignment for $\gamma\ge4$.
Then, Algorithm~\ref{alg:main} returns a hypergraph of expected size
$\EX(|\wt E|) \leq2\lambda\gamma n$.
\end{lemma}

\begin{proof}
Each hyperedge $e$ contributes $p(e)\le w(e)\cdot R^{\max}(e)\cdot\lambda$ to $\EX(|\wt E|)$,
and it thus suffices to bound
\begin{align} \label{eq:hyper-eff-res-sum}
  \sum_{e\in E}w(e)\cdot R^{\max}(e)\le2\gamma\cdot n.
\end{align}

We proceed to prove this inequality.
For each $e\in E$, let us partition $F_e$ into two groups,
$F_e^{(1)} = \set{f\in F_e \mid \gamma\cdot R_G(f)\ge R^{\max}(e)}$
and the remaining edges $F_e^{(2)} = F_e\setminus F_e^{(1)}$.
By Definition~\ref{def:weighted-approx-balanced},
these remaining edges $f\in F_e^{(2)}$ satisfy $\w(f)\le w(e)/n^2$.
Then, by Fact~\ref{fact:eff-res-sum}, we have

\begin{align*}
  n-1
  &=\sum_{f\in F}\w(f)\cdot R_G(f)
  =\sum_{e\in E}\sum_{f\in F_e}\w(f)\cdot R_G(f)
  \ge\sum_{e\in E}\sum_{f\in F_e^{(1)}}\w(f)\cdot\frac1\gamma\cdot R^{\max}(e) ,
\end{align*}
where the second equality is due to the fact that $F$ is the multiset $\bigcup_{e \in E} F_e$ by Definition~\ref{def:balanced}.
For the inner summation, we can bound
\begin{align*}
  \sum_{f\in F_e^{(1)}}\w(f)
  \ge \sum_{f\in F_e}\w(f) - \sum_{f\in F_e^{(2)}} \frac{w(e)}{n^2}
  \ge w(e) - \frac12\cdot w(e)
  = \frac12\cdot w(e),
\end{align*}
and altogether we obtain
$\sum_{e\in E}w(e)\cdot R^{\max}(e)\le2\gamma\cdot (n-1)$, which completes the proof. 
\end{proof}

Setting $\gamma=4$, for example, thus produces a linear-size output.
We now prove that the output is indeed a spectral sparsifier of $H$.

\begin{lemma}\label{lem:main}
Let $H=(V,E, w)$ be a weighted hypergraph and let $G=(V,F,\w)$ be its $1/n^2$-approximate $\gamma$-balanced weight assignment for some \emph{constant} $\gamma\ge4$. Then executing Algorithm~\ref{alg:main} on $H$, $G$, $1/n\le\epsilon<1$, and $\lambda=O(\log^3(n)/\epsilon^4)$,
returns with probability at least $1-O(\log(n)/n)$
an $\epsilon$-spectral sparsifier of $H$.
\end{lemma}

\begin{proof}
We proceed similarly to our proof of Theorem~\ref{thm:ordinary-sparsification} in Section~\ref{sec:warm-up}.
By Definition~\ref{def:hyper-sparsifier},
we must prove that for every $x\in\mathbb R^V$,
\begin{equation}\label{eq:hyper-main}
  \left|Q_H(x)-Q_{\wt H}(x)\right|\le\epsilon\cdot Q_H(x),
\end{equation}
Since Equation~\eqref{eq:hyper-main} is invariant to scaling,
we may assume without loss of generality that $x^\top L_G x = 1$,
and we denote the set of such vectors by $S^G\subseteq\mathbb R^V$. Notice that $S^G$ is defined with respect to $L_G$ and not $Q_H$,
however it implies that $Q_H(x)\ge1$ because for all $x\in S^G$, 
	\begin{align}\label{eq:hyper-Q-bound}
		Q_H(x)&=\sum_{e\in E}w(e)\cdot\max_{u^*,v^*\in e}(x_{u^*}-x_{v^*})^2=\sum_{e\in E}\sum_{f\in F_e}\w(f) \cdot \max_{u^*,v^*\in e}(x_{u^*}-x_{v^*})^2\\
		&\ge\sum_{e\in E}\sum_{f=(u,v)\in F_e}\w(f)\cdot(x_u-x_v)^2\label{eq:G-vs-H}=\sum_{f=(u,v)\in F}\w(f)\cdot(x_u-x_v)^2\nonumber=x^\top L_G x\nonumber=1.\nonumber
	\end{align}

For any $E'\subseteq E$, we denote the quadratic form $Q(x)$ restricted only to the hyperedges $E'$ in $H$ by $Q_{E'}(x)$,
and similarly in the hypergraph $\wt H$ by $\wt Q_{E'}(x)$.
If $E'=E$, we omit the subscript.
It is clear by the construction of $\wt H$ that
$$\mathbb E\left(\wt Q(x)\right)=Q(x).$$
Therefore, for any specific vector $x$, Equation~\eqref{eq:hyper-main} holds by Chernoff bound (Theorem~\ref{thm:chernoff}). In order to prove it simultaneously for all $x\in S^G$, we again use progressively finer and finer roundings of $x$, as guaranteed by Lemma~\ref{lem:rounding}.

Indeed, fix a sequence of the rounding functions $\varphi_i$ guaranteed by Lemma~\ref{lem:rounding} for $i=1,\ldots,I:=\log_2(14\gamma n/\epsilon)\le3\log n$ (since $\gamma$ is a fixed constant and $n$ is sufficiently large),
and denote the sequence of rounded vectors for $x\in S^G$ by $x^{(i)}=\varphi_i(x)$.
We define for each $x^{(i)}$ the set of hyperedges 
$$E_i':=\left\{e\in E\;\middle|\;\max_{u,v\in e}\left(x_u^{(i)}-x_v^{(i)}\right)^2\ge2^{-i}\cdot R^{\max}(e)\right\},$$
where by definition $R^{\max}(e) = \max_{f\in F_e}R_G(f)$.
This set $E'_i$ is designed such that the rounding $\varphi_i$ will conserve (to within a small multiplicative error) $Q_{\{e\}}(x)$ for any $e\in E_i'$ (we will see a proof of this fact later on).
We can now define a partition of $E$ based on the sets $E_i'$.
First, recall that $p(e)=\min(1,\lambda\cdot w(e)\cdot R^{\max}(e))$,
and let the base case be $E_0 = E_0' := \set{e\in E \mid p(e) = 1}$.
Then for each $i=1,\ldots,I$, let
$$
E_i:=E_i'\setminus\bigcup_{j=0}^{i-1}E_j'.
$$
Finally, let $E_{I+1}:=E\setminus\bigcup_{i=0}^I E_i'$.

	We begin with two useful claims that will help us bridge the gap between the ordinary graph and the hypergraph settings. First, we extend the second guarantee of Lemma~\ref{lem:rounding} to hyperedges.

\begin{claim}\label{claim:hyper-rounding}
For all $x\in S^G$ and $e\in E$ such that
$$\max\left(\max_{u,v\in e}(x_u-x_v)^2,\max_{u,v\in e}(x_u^{(i)}-x_v^{(i)})^2\right)\ge2^{-i}\cdot R^{\max}(e),$$
we have
\begin{equation}\label{eq:hyper-rounding}
\max_{u,v\in e}(x_u-x_v)^2 \in \left(1 \pm \frac\epsilon7\right)\cdot \max_{u,v\in e}(x_u^{(i)}-x_v^{(i)})^2.
\end{equation}
\end{claim}
\begin{proof}
We focus on the case when $\max_{u,v\in e}(x_u^{(i)}-x_v^{(i)})^2\ge2^{-i}\cdot R^{\max}(e)$. The case when $\max_{u,v\in e}(x_u-x_v)^2\ge2^{-i}\cdot R^{\max}(e)$ follows by an identical argument.
	
For one direction,
let the pair $u^*,v^*\in e$ maximize $(x^{(i)}_{u}-x^{(i)}_{v})^2$.
Now
$$(x^{(i)}_{u^*}-x^{(i)}_{v^*})^2=\max_{u,v\in E}(x^{(i)}_u-x^{(i)}_v)^2\ge2^{-i}\cdot R^{\max}(e)=2^{-i}\cdot\max_{f\in F_e}R_G(f)\ge2^{-i}\cdot R_G(u^*,v^*),$$
hence the second guarantee of Lemma~\ref{lem:rounding} holds for $(u^*,v^*)$,
and consequently
$$(x^{(i)}_{u^*}-x^{(i)}_{v^*})^2\le\left(1-\frac\epsilon7\right)^{-1}\cdot(x_{u^*}-x_{v^*})^2\le\left(1-\frac\epsilon7\right)^{-1}\cdot\max_{u,v\in e}(x_u-x_v)^2.$$

The other direction follows by a similar argument, but is slightly more complicated. The asymmetry is due to our assumption (that $\max_{u,v\in e}(x_u^{(i)}-x_v^{(i)})\ge2^{-i}\cdot R^{\max}(e)$) being in terms of $x^{(i)}$ instead of $x$.

Let $u^*,v^*$ be a vertex pair $u,v\in e$ that maximizes $(x_u-x_v)^2$. Here we distinguish two cases: If $(x_{u^*}-x_{v^*})^2<2^{-i}\cdot R^{\max}(e)$, we are immediately done, since
$$\max_{u,v\in e}(x_u-x_v)^2=(x_{u^*}-x_{v^*})^2\le2^{-i}\cdot R^{\max}(e)\le\max_{u,v\in e}(x_u^{(i)}-x_v^{(i)})^2.$$

On the other hand, if $(x_{u^*}-x_{v^*})^2\ge2^{-i}\cdot R^{\max}(e)$, we proceed identically to the first half of the proof:
$$(x_{u^*}-x_{v^*})^2\ge2^{-i}\cdot R^{\max}(e)=2^{-i}\cdot\max_{f\in F_e}R_G(f)\ge2^{-i}\cdot R_G(u^*,v^*).$$
Therefore, the second guarantee of Lemma~\ref{lem:rounding} holds for $(u^*,v^*)$ in $G$, and in particular
\[
(x_{u^*}-x_{v^*})^2\le\left(1+\frac\epsilon7\right)\cdot(x_{u^*}^{(i)}-x_{v^*}^{(i)})^2\le\left(1+\frac\epsilon7\right)\cdot\max_{u,v\in e}(x_u^{(i)}-x_v^{(i)})^2.
\qedhere
\]
\end{proof}

We thus obtain an analogue to Corollary~\ref{cor:rounding} from Section~\ref{sec:warm-up}.

\begin{corollary}\label{cor:hyper-rounding}
For every edge $e\in E_i$,
$$Q_{\{e\}}(x)\in \left(1 \pm \frac\epsilon7\right)\cdot Q_{\{e\}}(x^{(i)}).$$
The same holds also for $\wt Q_{\{e\}}$ (because it is a multiple of $Q_{\{e\}}$).
\end{corollary}

\begin{proof}
	We have that $e\in E_i\subseteq E_i'$ and hence $\max_{u,v\in e}(x_u^{(i)}-x_v^{(i)})^2\ge2^{-i}\cdot R^{\max}(e)$. We can therefore apply Claim~\ref{claim:hyper-rounding} and scale up by $w(e)$ to get the desired result.
\end{proof}

Next, we prove an analogue to Claim~\ref{cl:power-bound} from Section~\ref{sec:warm-up}.

\begin{claim}\label{cl:hyper-power-bound}
For all $i\in[I]$ and $e\in E_i$, we have
$$\max_{u,v\in e}(x_u^{(i)}-x_v^{(i)})^2\le3\cdot2^{-i}\cdot R^{\max}(e) .$$
\end{claim}

\begin{proof}
The proof proceeds nearly identically to that of Claim~\ref{cl:power-bound}.

The condition of Claim~\ref{claim:hyper-rounding} holds for $e\in E_i$ and hence we have $\max_{u,v\in e}(x_u^{(i)}-x_v^{(i)})^2\le(1-\epsilon/7)^{-1}\cdot\max_{u,v\in e}(x_u-x_v)^2$.

Consider first the case $i=1$.
Let $u^*,v^*\in e$ maximize $(x_u-x_v)^2$.
Then by Fact~\ref{fact:eff-res} and since $x\in S^G$, we have
$\max_{u,v\in e}(x_u-x_v)^2=(x_{u^*}-x_{v^*})^2\le R_G(u^*,v^*)\cdot x^\top L_G x = R_G(u^*,v^*)$,
and we indeed get
\begin{equation*}
\begin{split}
\max_{u,v\in e}(x_u^{(i)}-x_v^{(i)})^2&\leq \max_{u,v\in e}(x_u-x_v)^2\cdot \left(1-\frac{\epsilon}{7}\right)^{-1}\\
&\le R_G(u^*,v^*)\cdot \left(1-\frac{\epsilon}{7}\right)^{-1}\\
&\le R^{\max}(e)\cdot \left(1-\frac{\epsilon}{7}\right)^{-1}\\
&\le3\cdot 2^{-1}\cdot R^{\max}(e).
\end{split}
\end{equation*}

Now consider $i>1$, and suppose towards contradiction that
$\max_{u,v\in e} (x_u^{(i)}-x_v^{(i)})^2>3\cdot2^{-i}\cdot R^{\max}(e)$.
Notice that since $\max_{u,v\in e}(x_u-x_v)^2\ge2^{-i+1}\cdot R^{\max}(e)$, Claim~\ref{claim:hyper-rounding} still applies to $e$ with respect to $\varphi_{i-1}$.
Thus

\begin{equation*}
\begin{split}
\max_{u,v\in e}(x_u^{(i-1)}-x_v^{(i-1)})^2&\ge\max_{u,v\in e}(x_u-x_v)^2\cdot \left(1+\frac{\epsilon}{7}\right)^{-1}\\
&\geq \max_{u,v\in e}(x_u^{(i)}-x_v^{(i)})^2\cdot \left(1-\frac{\epsilon}{7}\right)\cdot \left(1+\frac{\epsilon}{7}\right)^{-1}\\
&\ge2^{-i+1}\cdot R^{\max}(e).
\end{split}
\end{equation*}

This implies that $e\in E_{i-1}'$,
which contradicts the assumption $e\in E_i$. 
\end{proof}

We consider each group of hyperedges $E_i$ separately,
and prove that the quadratic form is well approximated on this subset with high probability, that is
$$\wt Q_{E_i}(x)\approx Q_{E_i}(x).$$

\paragraph{Hyperedges in $E_0$.}
By definition, for all $e\in E_0$, we have $p(e)=1$.
Therefore, $\wt Q_{E_0}(x) = Q_{E_0}(x)$ deterministically.

\paragraph{Hyperedges in $E_i$ for $i\in[I]$.} Similarly to the case of ordinary graphs in Section~\ref{sec:warm-up}, we want to consider $x^{(i)}$ instead of $x$. We may do this due to Corollary~\ref{cor:hyper-rounding}.
Therefore, we can focus on proving that $Q_{E_i}(x^{(i)})\approx\wt Q_{E_i}(x^{(i)})$.

Consider a specific $i\in[I]$.
We wish to prove that with high probability,
simultaneously for all values of $x^{(i)}$ and $E^{(i)}$,
\begin{equation}\label{eq:hyper-mid-i-main}
  \left|\wt Q_{E_i}(x^{(i)})-Q_{E_i}(x^{(i)})\right|\le\frac{\epsilon\cdot Q(x)}{7I} .
\end{equation}
Since (by the first guarantee of Lemma~\ref{lem:rounding})
there are only finitely many values $(x^{(i)},E_i)$,
we can prove that Equation~\eqref{eq:hyper-mid-i-main} holds with high probability individually for each such pair,
and then use a union bound over all pairs.

The former can be done using Chernoff bounds (Theorem~\ref{thm:chernoff}); we begin with this.
Let us fix $x^{(i)}$ and $E_i$.
Recall that
$$\wt Q_{E_i}(x^{(i)})=\sum_{e\in E_i}\wt w(e)\cdot\max_{u,v\in e}(x_u^{(i)}-x_v^{(i)})^2$$
is the sum of independent random variables 
with expectation $\EX(\wt Q_{E_i}(x^{(i)})) = Q_{E_i}(x^{(i)})$
because $\mathbb E(\wt w(e))=w(e)$ by definition.
To apply Chernoff bounds (Theorem~\ref{thm:chernoff}), we need to set the variables $a$, $\delta$ and $\mu$.
We set $\delta:=\epsilon/(14I)$,
and using Corollary~\ref{cor:hyper-rounding},
we can bound
$$
  Q_{E_i}(x^{(i)})
  \le\left(1-\frac{\epsilon}{7}\right)^{-1}\cdot Q_{E_i}(x)\le2Q(x)
  =: \mu.
$$
(This is true for an arbitrary preimage $x\in\varphi_i^{-1}(x^{(i)})$.)
We also need to set $a$ as an upper bound on the largest possible value of any variable of the form $\wt w(e)\cdot\max_{u,v\in e}(x_u^{(i)}-x_v^{(i)})^2$ for $e\in E_i$.
Such a variable takes its maximum value when $e$ is sampled to the sparsifier $\wt H$, in which case $\wt w(e)=w(e)/p(e)=w(e)/(\lambda\cdot w(e)\cdot R^{\max}(e))=1/(\lambda\cdot R^{\max}(e))$, since $e\not\in E_0$.
Therefore, the random variable in question is upper bounded,
using Claim~\ref{cl:hyper-power-bound}, by
$$
  \max_{e\in E_i}\frac{\max_{u,v\in e}(x_u^{(i)}-x_v^{(i)})^2}{\lambda\cdot R^{\max}(e)}
  \le\frac{3\cdot2^{-i}}{\lambda}
 =: a .
$$

Finally, Theorem~\ref{thm:chernoff} implies
\begin{align*}
  \mathbb P\left(\left|Q_{E_i}(x^{(i)}-\wt Q_{E_i}(x^{(i)}))\right|\ge\frac{\epsilon\cdot Q(x)}{7I}\right)
   &\le2\exp\left(-\frac{\delta^2\mu}{3a}\right)\\
   &=2\exp\left(-\frac{\tfrac{\epsilon^2}{196I^2}\cdot Q(x)}{3\cdot3\cdot2^{-i} / \lambda}\right)\\
   &\le2\exp\left(-\frac{\epsilon^2\cdot2^{i}\cdot\lambda}{10000\log^2(n)}\right)\\
   &=2\exp\left(-\frac{2000C\log(n)\cdot2^i}{\epsilon^2}\right),
\end{align*}
where the last step is by setting $\lambda = 2\cdot10^7\cdot C\log^3(n)/\epsilon^4$.

We now turn to applying a union bound over all possible values of $(x^{(i)},E_i)$. Much like $F_i$ in the proof of Theorem~\ref{thm:ordinary-sparsification}, $E_i$ depends on all $E_j'$ for $j\le i$, which in turn depend on all $x^{(j)}$ for $j\le i$. Since $x^{(j)}=\varphi_j(x)$, by the first guarantee of Lemma~\ref{lem:rounding}, there are at most $\exp(800C\log(n)\cdot2^{j}/\epsilon^2)$ possible vectors $x^{(j)}$ (where $C$ is the absolute constant from Theorem~\ref{thm:rounding}). Therefore, the number of possible pairs $(x^{(i)},E_i)$ is at most
$$\prod_{j=1}^i\exp\left(\frac{800C\log(n)\cdot2^j}{\epsilon^2}\right)=\exp\left(\sum_{j=1}^i\frac{800C\log(n)\cdot2^j}{\epsilon^2}\right)\le\exp\left(\frac{1600C\log(n)\cdot2^i}{\epsilon^2}\right).$$

Putting together the Chernoff bound and the union bound, we get that
Equation~\eqref{eq:hyper-mid-i-main} holds simultaneously for all values of $(x^{(i)},E_i)$ except with probability at most
$$\exp\left(\frac{1600C\log(n)\cdot2^i}{\epsilon^2}\right)\cdot2\exp\left(-\frac{2000C\log(n)\cdot2^i}{\epsilon^2}\right)=2\exp\left(-\frac{400C\log(n)\cdot2^i}{\epsilon^2}\right)\le\frac1n.$$

\paragraph{Hyperedges in $E_{I+1}$.}
Recall that $I=\log_2(14\gamma\cdot n/\epsilon)$. First we show that for any hyperedge $e\in E_{I+1}$,
we have that $\max_{u,v\in e}(x_u-x_v)^2\le\epsilon\cdot R^{\max}(e)/(12\gamma\cdot n)$.
Indeed, suppose for contradiction that this is not the case,
and let $u^*,v^*\in e$ be a vertex pair that maximizes $(x_u-x_v)^2$.
Then, by the second guarantee of Lemma~\ref{lem:rounding}, we have that $(x^{(I)}_{u^*}-x_{v^*}^{(I)})^2\ge(x_{u^*}-x_{v^*})^2\cdot(1-\epsilon/7)\ge\epsilon\cdot R^{\max}(e)\cdot(1-\epsilon/7)/(12\gamma\cdot n)\ge\epsilon\cdot R^{\max}(e)/(14\gamma\cdot n)$.
Therefore, $e \in E_I'$, which contradicts $e\in E_{I+1}$.

We show that $|\wt Q_{E_{I+1}}-Q_{E_{I+1}}|$ is small by showing that each of the two terms is individually small.
First,
\begin{align*}
  Q_{E_{I+1}}(x)
  =\sum_{e\in E_{I+1}}w(e)\cdot\max_{u,v\in e}(x_u-x_v)^2
  \le\sum_{e\in E_{I+1}}w(e)\cdot\frac{\epsilon\cdot R^{\max}(e)}{12\gamma\cdot n}
  \le\frac{\epsilon}{12\gamma\cdot n}\cdot\sum_{e\in E}w(e)\cdot R^{\max}(e)
  \le\frac{\epsilon}{6},
\end{align*}
where the last inequality uses Equation~\eqref{eq:hyper-eff-res-sum}.
Second, we start similarly,
\begin{align*}
  \wt Q_{E_{I+1}}(x)
  =\sum_{e\in E_{I+1}}\wt w(e)\cdot\max_{u,v\in e}(x_u-x_v)^2
  \le\sum_{e\in E_{I+1}}\wt w(e)\cdot\frac{\epsilon\cdot R^{\max}(e)}{12\gamma\cdot n}
  \le\frac{\epsilon}{12\gamma\cdot n}\cdot\sum_{e\in E}\wt w(e)\cdot R^{\max}(e),
\end{align*}
and ideally we would like to show that $\sum\wt w(e)\cdot R^{\max}(e)\le4\gamma\cdot n$.
This is a random event, independent of the choice of $x$,
whose probability we can bound using the Chernoff bound (Theorem~\ref{thm:chernoff}).
Recall that $E_0 = \set{e\in E \mid p(e)=1}$,
and denote its complement by
$\overline E_0 := \set{e\in E \mid p(e)=\lambda\cdot w(e)\cdot R^{\max}(e)}$.
Since $\EX(\sum\wt w(e)\cdot R^{\max}(e))=\sum w(e)\cdot R^{\max}(e)\le2\gamma\cdot n$
(using Equation~\eqref{eq:hyper-eff-res-sum})
and $\sum_{e\in E_0}\wt w(e)\cdot R^{\max}(e)$ is deterministic by definition of $E_0$,
we have that
\begin{align*}
  \mathbb P\left(\sum_{e\in E}\wt w(e)\cdot R^{\max}(e)\ge4\gamma\cdot n\right)
  &\le\mathbb P\left(\left|\sum_{e\in E}\wt w(e)\cdot R^{\max}(e)-\mathbb E\left(\sum_{e\in E}\wt w(e)\cdot R^{\max}(e)\right)\right|\ge2\gamma\cdot n\right)\\
  &=\mathbb P\left(\left|\sum_{e\in\overline E_0}\wt w(e)\cdot R^{\max}(e)-\mathbb E\left(\sum_{e\in\overline E_0}\wt w(e)\cdot R^{\max}(e)\right)\right|\ge2\gamma\cdot n\right).
\end{align*}

We bound this by applying Theorem~\ref{thm:chernoff} for the independent random variables $\wt w(e)\cdot R^{\max}(e)$ where $e\in\overline E_0$.
The maximum value of such a variable occurs when $e$ is sampled, in which case it is exactly $\wt w(e)\cdot R^{\max}(e)=R^{\max}(e)\cdot w(e)/p(e)=1/\lambda\le1 =:a$.
Setting $\delta := 1$ and $\mu := 2\gamma\cdot n$
(we may do this due to Equation~\eqref{eq:hyper-eff-res-sum}),
we get
$$\mathbb P\left(\sum_{e\in E}\wt w(e)\cdot R^{\max}(e) \ge 4\gamma\cdot n\right)\le2\exp\left(-\frac{2\gamma\cdot n}{3}\right)\le\frac1n.$$
In conclusion, with probability at least $1-O(1/n)$,
\begin{equation}\label{eq:hyper-I+1-main}
  \left|\wt Q_{E_{I+1}}(x)-Q_{E_{I+1}}(x)\right|\le\frac{\epsilon}{2} .
\end{equation}

\paragraph{Putting everything together.}
The final part of the proof proceeds identically to that of Theorem~\ref{thm:ordinary-sparsification}.
We have seen that Equation~\eqref{eq:hyper-mid-i-main} holds simultaneously for all $i$ and $(x^{(i)},E_i)$,
as well as Equation~\eqref{eq:hyper-I+1-main} holds with probability at least $1-O(\log(n)/n)$.
Conditioning on these events, we can deduce Equation~\eqref{eq:hyper-main},
thus concluding the proof of Lemma~\ref{lem:main}.

For completeness we repeat the derivation:
\begin{align*}
	\left|Q_{E_i}(x)-\wt Q_{E_i}(x)\right|
	\le&\left|Q_{E_i}(x)-Q_{E_i}(x^{(i)})\right|+\left|Q_{E_i}(x^{(i)})-\wt Q_{E_i}(x^{(i)})\right|+\left|\wt Q_{E_i}(x^{(i)})-\wt Q_{E_i}(x)\right|\\
	\le&\frac{\epsilon}{7}\cdot Q_{E_i}(x^{(i)}) + \left|Q_{E_i}(x^{(i)})-\wt Q_{E_i}(x^{(i)})\right| + \frac{\epsilon}{7}\cdot\wt Q_{E_i}(x^{(i)}) \tag{by Corollary~\ref{cor:hyper-rounding}}\\
	\le& \frac{2\epsilon}{7}\cdot Q_{E_i}(x^{(i)}) + \left(1+\frac{\epsilon}{7}\right)\cdot\left|Q_{E_i}(x^{(i)})-\wt Q_{E_i}(x^{(i)})\right|\\
	\le& \frac{2\epsilon}{7}\cdot \left(1-\frac{\epsilon}{7}\right)^{-1}\cdot Q_{E_i}(x)+\left(1+\frac{\epsilon}{7}\right)\cdot\frac{\epsilon\cdot Q(x)}{7I} \tag{by Corollary~\ref{cor:hyper-rounding} and Equation~\eqref{eq:hyper-mid-i-main}}\\
	\le& \frac{\epsilon}{3}\cdot Q_{E_i}(x)+\frac{\epsilon\cdot Q(x)}{6I}.
\end{align*}
Therefore, we have
\begin{align*}
	\left|Q(x)-\wt Q(x)\right|\le&\left|Q_{E_0}(x)-\wt Q_{E_0}(x)\right| +\sum_{i=1}^I \left| Q_{E_i}(x)-\wt Q_{E_i}(x)\right|+\left|Q_{E_{I+1}}(x)-\wt Q_{E_{I+1}}(x)\right|\\
	\le&0+\sum_{i=1}^I\left(\frac{\epsilon}{3}\cdot Q_{E_i}(x)+\frac{\epsilon\cdot Q(x)}{6I}\right)+\frac{\epsilon}{2} \tag{by Equation~\eqref{eq:hyper-I+1-main}}\\
	\le&\frac{\epsilon}{3}\cdot Q(x)+\frac{\epsilon}{6}\cdot Q(x)+\frac{\epsilon}{2}\\
	\le&\epsilon\cdot Q(x), \tag{by Equation~\eqref{eq:hyper-Q-bound}}
\end{align*}
as claimed.
\end{proof}
Combining Theorem~\ref{thm:greedy-approx-balanced} and Lemmas~\ref{lem:hypergraph_size} and~\ref{lem:main}, we get the first part of Theorem~\ref{thm:main}.


%% file: 600-speedup.tex

\section{Nearly Optimal Speed-Up}\label{sec:speed-up}

In the previous section, we have proved the existence of nearly linear sized spectral sparsifiers for hypergraphs. We have also provided an method for constructing such sparsifiers: we construct an approximate balanced weight assignment of the input hypergraph using Algorithm~\ref{alg:greedy-approx-balanced}, and then construct a sparsifier using Algorithm~\ref{alg:main}. However, the running time of Algorithm~\ref{alg:greedy-approx-balanced} on an unweighted hypergraph is $m\cdot\poly{n}$, which is large; in this section we improve this to $\wt{O}(mr) + \poly{n}$, that is, we prove the second part of Theorem~\ref{thm:main}.
As we mentioned in the introduction, with a small modification, this leads to an algorithm with time complexity $\wt{O}(\sum_{e \in E}|e| + \poly{n})$ that constructs an $\epsilon$-spectral-sparsifier of nearly linear size.
This running time is optimal to within polylogarithmic factors in $n$, unless the size of the input hypergraph is polynomially small in $n$.

Our algorithm consists of two steps. First we apply the algorithm of~\cite{Bansal2019}, which---with small modifications---can be shown to run in $\wt{O}(mr) + \poly{n}$ time, but which produces a larger spectral sparsifier. We then aim to sparsify the resulting weighted hypergraph using our methods. Unfortunately, even though the resulting hypergraph has only polynomially many hyperedges (in $n$), the range of edge weights may still be exponential, meaning that Algorithm~\ref{alg:greedy-approx-balanced} is not efficient for finding an approximate balanced weight assignment on it (recall Theorem~\ref{thm:greedy-approx-balanced}). We propose a variation of Algorithm~\ref{alg:greedy-approx-balanced} adapted for this setting which runs in polynomial time.

\subsection{Fast Algorithm for Constructing Polynomial-Sized Sparsifiers}

In this section, we recall and slightly modify the algorithm of~\cite{Bansal2019} for producing polynomial-sized spectral sparsifiers for hypergraphs.

\begin{definition}\label{def:associated-graph}
	For a weighted hypergraph $H=(V,E)$, let $G(H)$ denote the ``associated  graph'' of $H$, which is defined as follows: Replace each hyperedge $e$ of $E$ with a clique of uniform weight $w(e)$ on the support of $e$. (Note that this may produce parallel edges.)
\end{definition}

\begin{theorem}[\cite{Bansal2019}]\label{thm:ola-sparsifiaction}
	Let $H=(V,E,w)$ be a hypergraph where all hyperedges have size between $r/2$ and $r$. Then, for some absolute constant $c$, the following process produces an $\epsilon$-spectral sparsifier for $H$ with probability at least $1-1/n$: Let $G(H)$ be the associated graph of $H$. For each hyperedge $e\in E$, let
	$$R^{\max}(e)=\max_{u,v \in e}R_{G(H)}(u,v).$$
	Sample each hyperedge $e$ independently with some probability
	\begin{equation}\label{eq:bansal-et-al}
	p(e)\ge\min\left(1, \frac{w(e)\cdot R^{\max}(e)\cdot r^4\log n}{c\epsilon^2}\right),
	\end{equation}
	and if sampled give it weight $\wt w(e)=w(e)/p(e)$.
\end{theorem}

\begin{remark}
	In fact, in~\cite{Bansal2019}, the result is stated slightly less generally, for unweighted hypergraphs, and without allowing oversampling (that is $p(e)$ is set \emph{exactly} to the right hand side in Equation~\ref{eq:bansal-et-al}, instead of being lower bounded by it). However, this version holds by an identical proof.
\end{remark}

The trivial implementation of this takes time $\Omega(mr^2)$ in general for two different reasons: First, it takes $\Omega(mr^2)$ time to replace each of the $m$ hyperedges with a clique of size $r^2$. Second, it takes $\Omega(mr^2)$ time to find $\max_{u,v\in e}R_{G(H)}(u, v)$ for all $m$ hyperedges. However, with some simple tricks in the implementation both bottlenecks can be avoided to reduce the running time to $\wt O(mr+\poly{n})$.

To achieve this, we first replace cliques in the associated graph of the input hypergraph with sparse graphs guaranteed by the following fact:

\begin{fact}\label{fact:g-star}
	It follows by Theorem~\ref{thm:ordinary-sparsification} that for every $r$, there exists a (weighted) graph $G^*_r$ with $r$ vertices and $\wt O(r)$ edges such that $G^*_r$ is a $1/2$-spectral sparsifier to the $r$-clique.
\end{fact}
Second, to calculate $R^{\max}(e)$ approximately, we do not take the maximum over all pairs of vertices in $e$, but only over $(u_0,v)$ for all $v\in e$ but for some fixed $u_0$. Since effective resistance is a metric (see Fact~\ref{fact:eff-res-metric}), this provides a $2$-approximation to $R^{\max}(e)$ by triangle inequality.

\begin{algorithm}[H]
	\caption{Fast algorithm for computing a polynomial-sized spectral sparsifier for an approximately uniform hypergraph.}\label{alg:uniform-sparsification}
	\begin{algorithmic}[1]
		\Procedure{UniformSparsification}{$H=(V,E,w),r,\epsilon$}
		\State $G=(V,F,\w)\gets(V,\emptyset,0)$
		\ForAll{$e\in E$}
			\State Add a copy of $G^*_{|e|}$ to $G$, supported on $e$, with weights scaled up by $w(e)$\label{line:bottleneck-1}
		\EndFor
		\State Calculate $R_G(u,v)$ for all $u,v\in V$
		\State $\wt H=(V,\wt E,\wt w)\gets(V,\emptyset,0)$
		\ForAll{$e\in E$}
			\State $u_0\gets$ an arbitrary vertex in $e$
			\State $\wt R^{\max}(e)\gets\max_{v\in e}R_{G}(u_0,v)$\label{line:bottleneck-2}
			\State $p(e)\gets\min\left(1, \frac{4w(e)\cdot \wt R^{\max}(e)\cdot r^4\log n}{c\epsilon^2}\right)$
			\State Add $e$ to $\wt{E}$ with probability $p(e)$ with weight $\wt w(e)\gets w(e)/p(e)$
		\EndFor
		\State \Return $\wt H$
		\EndProcedure
	\end{algorithmic}
\end{algorithm}

\begin{lemma}\label{lem:uniform-sparsification}
	If the input hypergraph $H=(V,E,w)$ only has hyperedges of size between $r/2$ and $r$, then Algorithm~\ref{alg:uniform-sparsification} runs in $\wt O(mr)+\poly{n}$ time, returning an $\epsilon$-spectral sparsifier to $H$ with probability at least $1-1/n$. Furthermore, the output has at most $4n^5\log n/(c\epsilon^2)$ hyperedges in expectation, where $c$ is the absolute constant from Theorem~\ref{thm:ola-sparsifiaction}.
\end{lemma}

\begin{proof}
	It takes $\wt{O}(mr)$ time to construct the graph $G$ (here Line~\ref{line:bottleneck-1} takes $\wt O(r)$ time) and its Laplacian. All pairs effective resistances can then be calculated in $\poly{n}$ time and stored in a table, resulting in an $O(r)$ time bound for the calculation of $\wt R^{\max}(e)$ in Line~\ref{line:bottleneck-2}.
	
	To show that the output is an $\epsilon$-spectral sparsifier of $H$ with high probability, it suffices to verify Equation~\ref{eq:bansal-et-al} of Theorem~\ref{thm:ola-sparsifiaction}, ie. that $p(e)$ is always at least $$w(e)\cdot\max_{u,v\in e}R_{G(H)}(u,v)\cdot \frac{r^4\log n}{c\epsilon^2}.$$
	For this, it suffices to show that $\wt R^{\max}(e)$ (as defined in Line~\ref{line:bottleneck-2} of Algorithm~\ref{alg:uniform-sparsification}) is at least $\max_{u,v\in e}\allowbreak R_{G(H)}(u, v)/4$. Since $w(e)\cdot G_{|e|}^*$ is a $1/2$-spectral sparsifier of the clique on $e$ (of uniform edge weight $w(e)$), it follows by the additivity of Laplacians that $G$ from Algorithm~\ref{alg:uniform-sparsification} is a $1/2$-spectral sparsifier of $G(H)$ from Definition~\ref{def:associated-graph}. Therefore, $R_G(u, v)\ge R_{G(H)}(u,v)/2$ for all $u, v\in V$ (recall Definition~\ref{def:eff-res}). Finally, since $R_G$ is a metric on $V$ (by Fact~\ref{fact:eff-res-metric}),
	$$\max_{v\in e}R_G(u_0,v)\ge\max_{u,v\in e}R_G(u,v)/2.$$
	Indeed, if $(u^*,v^*)$ maximizes $R_G(u,v)$, then by triangle inequality $R_G(u^*,v^*)\le R_G(u^*,u_0) + R_G(u_0,v^*)$; one of the latter two must be at least $R_G(u^*,v^*)/2$.
	
	This concludes the proof of correctness; we must finally prove the upper bound on the expected size of the output $\wt H$.
	Since $r^4\le n^4$, it suffices to show that
	$\sum_{e\in E}w(e)\cdot\wt R^{\max}(e)\le n$. First note that
	$$\sum_{e\in E}w(e)\cdot\wt R^{\max}(e)\le\sum_{e\in E}w(e)\cdot\max_{u,v\in e}R_G(u,v)\le2\sum_{e\in E}w(e)\cdot R_{G(H)}(u,v),$$
	since $G$ is a $1/2$-spectral sparsifier of $G(H)$. Then, since the weight of $(u,v)$ in $G(H)$ is exactly $w(e)$ by definition, we have that
	$$\sum_{e\in E}w(e)\cdot\wt R^{\max}(e)\le\sum_{f\in F(H)}\w_{G(H)}(f)\cdot R_{G(H)}(f)=n-1,$$
	by Fact~\ref{fact:eff-res-sum}, where $F(H)$ denotes the edge-set of $G(H)$ and $\w_{G(H)}$ denotes the weight function of $G(H)$.
\end{proof}

Algorithm~\ref{alg:uniform-sparsification} only works under the constraint that the input hypergraph is approximately uniform -- that is, the sizes of hyperedges all fall into the range $[r/2,r]$. This is easily circumvented, however: Given an arbitrary input hypergraph, one can simply partition the hyperedges into a logarithmic number of parts based on cardinality. We then sparsify the parts, and combine the resulting sparsifiers, losing only a $\log n$ factor in size.

We formalize this in the following algorithm and corollary.

\begin{algorithm}[H]
	\caption{Fast algorithm for computing a polynomial-sized spectral sparsifier for an arbitrary hypergraph.}\label{alg:poly-sparsification}
	\begin{algorithmic}[1]
		\Procedure{PolynomialSizeSparsification}{$H=(V,E,w),\epsilon$}
		\For{$i$ from $1$ to $\log n$}
			\State $E_i\gets\{e\in E\big||e|\in[2^i,2^{i+1})\}$
			\State $H_i\gets(V,E_i,w)$
			\State $\wt H_i\gets\textsc{UniformSparsification}(H_i,2^{i+1},\epsilon)$\Comment{Algorithm~\ref{alg:uniform-sparsification}}
		\EndFor
		\State \Return $\wt H\gets\cup_{i=1}^{\log n}\wt H_i$
		\EndProcedure
	\end{algorithmic}
\end{algorithm}
Here the final line means that we take all hyperedges (with associated weights) from all of $\wt H_1,\ldots,\wt H_{\log n}$.

\begin{corollary}\label{cor:poly-sparsification}
	Algorithm~\ref{alg:poly-sparsification} runs in $\wt O(mr)+\poly{n}$ time, returning an $\epsilon$-spectral sprasifier of the input $H$ with probability at least $1-\log n/n$. Furthermore, the output has at most $4n^5\log^2(n)/(c\epsilon^2)$ hyperedges in expectation, where $c$ is the absolute constant from Theorem~\ref{thm:ola-sparsifiaction}.
\end{corollary}

\subsection{Even Faster Construction for $\gamma$-Balanced Weight Assignments}\label{sec:separated}

It is difficult to get a stronger bound on the number of rounds of Algorithm~\ref{alg:greedy-approx-balanced}, than that of Theorem~\ref{thm:greedy-approx-balanced}, at least in its full generality.
Instead, here we define a specific class of hypergraphs for which a better algorithm exists.

\begin{definition}\label{def:separated}
A weighted hypergraph $H=(V,E,w)$ is called \emph{$(\alpha,\beta)$-separated}
for parameters $\alpha\geq 1$ and $\beta\geq 1$ if the hyperedge set $E$ is partitioned into $E_1,\ldots,E_\ell$,
satisfying the two requirements:
\begin{align*}
  &\text{for all~}i\in \{1,2,\ldots, \ell\},
  &\max_{e\in E_i}w(e)\le\alpha\cdot\min_{e\in E_i}w(e),
  \\
  &\text{for all~}i, j\in \{1,2,\ldots, \ell\}, i<j,
  &\min_{e\in E_i}w(e)\ge\beta\cdot\max_{e\in E_j}w(e).
\end{align*}
\end{definition}

Our next algorithm exploits the structure of separated hypergraphs in order to more efficiently construct approximate balanced weight assignments on them. Intuitively, one can think of the different weight classes in separated hypergraphs as only weakly interacting, which is the source of our speedup. In particular, it is important to note that we will get a speedup for hypergraphs where the total number of hyperedges is small in comparison to the amount of separation between the weight classes (i.e., hypergraphs that are produced by the sparsification procedure from the previous section)---this is what ultimately ensures that the different weight classes only interact in a limited manner.

\begin{algorithm}[H]
	\caption{Computing an $\eta$-approximate $\gamma$-balanced weight assignment for an $(\alpha,\beta)$-separated hypergraph.}\label{alg:separated-approx-balanced}
	\begin{algorithmic}[1]
		\Procedure{SeparatedApproxBalancing}{$H=(V,E,w)$,$(E_i)_{i=1}^\ell,\alpha,\beta,\gamma,\eta$}
		\State For all $e\in E$ and for all $f\in F_e$ initialize $\w(f) \gets w(e)/\binom{|e|}{2}$
		\State $G\gets(V,\bigcup F_e,\w)$
		\For{$i = 1,\ldots,\ell$}
		\While{there exists $e\in E_i$ violating the conditions of $\eta$-approximate $\gamma$-balancedness}
		\State Select such $e\in E_i$ and $f,g\in F_e$, such that $R_G(f)>\gamma\cdot R_G(g)$ and $\w(g)>\eta\cdot w(e)$
		\State $\lambda\gets\min\left(\w(g),1/(2R_G(g))\right)$
		\State $\w(f)\gets\w(f)+\lambda$
		\State $\w(g)\gets\w(g)-\lambda$
		\EndWhile
		\EndFor
		\State \Return $G$.\label{line:separated-return}
		\EndProcedure
	\end{algorithmic}
\end{algorithm}

In fact, Algorithm~\ref{alg:separated-approx-balanced} is very similar to Algorithm~\ref{alg:greedy-approx-balanced}. However, it corrects discrepancies in heavier hyperedges first, and once a category of hyperedges ($E_i$) has been corrected, the algorithm never goes back to it, \emph{not even if the approximate balancedness becomes violated}. For this reason, the resulting output, $G$, will not necessarily be $\gamma$-balanced. However, the structure of separated hypergraphs will allow us to bound the number of rounds using $\alpha$, instead of the global aspect ratio of weights, which could be much larger.

The crucial property of separated hypergraphs (and their weight assignments) which we exploit is that the heavier edges have a much greater influence on the effective resistance of a vertex pair than the lighter edges. More specifically, for some $i\in[\ell]$ we can define the subgraph $G_+$ containing only edges in $F_e$ for $e\in E_j$ where $j\le i$. Then, when calculating the effective resistance $R_G(u,v)$---under certain circumstances---we can simply calculate $R_{G_+}(u,v)$ instead, ignoring the contribution of the remaining edges. We formalize this in the following lemma and its corollary.

\begin{lemma}\label{lem:separated-eff-res}
	Let $G_+=(V,E_+,\w_+)$ and $G_-=(V,E_-,\w_-)$ be two weighted ordinary graphs on the same vertex set. Let the total weight of all edges in $G_-$ be bounded by $\zeta$. Let $G=G_+\cup G_-$ be the union of the two graphs, that is $G=(V,E_+\cup E_-,\w_+\cup\w_-)$. Then, for any vertex pair $(u,v)\in V$ in the same connected component of $G_+$, we can bound the effective resistance $R_G(u,v)$ in terms of $R_{G_+}$ as follows:
	$$\frac1{R_G(u,v)}\le\frac1{R_{G_+}(u,v)}+\zeta.$$
\end{lemma}

\begin{proof}
	To prove this inequality, we use the alternate definition of effective resistance from Fact~\ref{fact:eff-res}. That is
	$$R_{G_+}(u,v)=\max_{x\in\mathbb R^V}\frac{(x_u-x_v)^2}{x^\top L_{G_+}x}.$$
	Let $x^*$ maximize the above formula. We may assume without loss of generality that $x_u^*=0$ and $x_v^*=1$ since the formula is both shift and scale invariant. From this it follows that $x^{*\top}L_{G_+}x^*=1/R_{G_+}(u,v)$. We can further assume without loss of generality that $x_a^*\in[0,1]$ for all $a\in V$. Indeed, let the connected component of $u$ and $v$ in $G_+$ be $C\subseteq V$. Then for $a\in V\setminus C$, we may simply assume that $x^*_a$ is always $0$. To see that $x_a^*\in[0,1]$ for all $a\in C$, suppose for contradiction that $\max_{a\in C}x_a^*>1$. Let the highest value of $x^*$ in $C$ be $\mu_1>1$ and the second highest distinct value be $\mu_2$. Then one can change the $x^*$-value of vertices from $\mu_1$ to $\mu_2$, thereby strictly increasing the value of $(x^*_u-x^*_v)^2/x^{*\top}L_{G_+}x^*$. This is a contradiction, since $x^*$ maximizes the formula by definition. An identical argument by contradiction rules out that $\min_{a\in C}x_a^*<0$.
	
	We can now upper bound $1/R_G(u,v)$ using the same alternate definition from Fact~\ref{fact:eff-res}:
	\begin{align*}
	R_G(u,v)=&\max_{x\in\mathbb R^V}\frac{(x_u-x_v)^2}{x^\top L_G x}\ge\frac{(x_u^*-x_v^*)^2}{x^{*\top}L_G x^*}=\frac{1}{1/R_{G_+}(u,v)+x^{*\top}L_{G_-}x^*}.
	\end{align*}
	Now
	$$x^{*\top}L_{G_-}x^*=\sum_{(a,b)\in E_-}\w_-(a,b)\cdot(x_a^*-x_b^*)^2\le\sum_{(a,b) \in E_-}\w_-(a,b)\le\zeta.$$
	Therefore
	$$\frac1{R_G(u,v)}\le\frac1{R_{G_+}(u,v)}+\zeta,$$
	as claimed.
\end{proof}

\begin{corollary}\label{cor:separated-eff-res}
	Let $G_+=(V,E_+,\w_+)$ and $G_-=(V,E_-,\w_-)$ be two weighted ordinary graphs on the same vertex set. Let the total weight of all edges in $G_-$ be bounded by $\zeta$. Let $G=G_+\cup G_-$ be the union of the two graphs, that is $G=(V,E_+\cup E_-,\w_+\cup\w_-)$. Then for any vertex pair $u,v\in V$, if $R_G(u,v)\le1/(5\zeta)$ then we can bound the effective resistance $R_G(u,v)$ in terms of $R_{G_+}(u, v)$ as follows:
	$$R_G(u,v)\ge\frac45\cdot R_{G_+}(u,v).$$
\end{corollary}

\begin{proof}
	We wish to apply Lemma~\ref{lem:separated-eff-res}, for which we must first show that $u$ and $v$ are in the same connected component in $G_+$. This is indeed the case: If $u$ and $v$ are in different connected components of $G_+$, we can use the alternate definition of effective resistance from Fact~\ref{fact:eff-res}:
	$$R_G(u,v)=\max_{x\in\mathbb R^V}\frac{(x_u-x_v)^2}{x^\top L_G x}\ge\frac{(x_u^0-x_v^0)^2}{x^{0\top}L_G x^0},$$
	where $x^0$ is $0$ on the connected component of $u$ in $G_+$, and $1$ everywhere else (including $v$). Now
	$$\frac{(x_u^0-x_v^0)^2}{x^{0\top}L_G x^0}=\frac{1}{x^{0\top}L_{G_+}x^0+x^{0\top}L_{G_-}x^0}=\frac{1}{x^{0\top}L_{G_-}x^0}=\frac{1}{\sum_{(a,b)\in E_-}\w_-(a,b)\cdot(x_a^0-x_b^0)^2}\ge\frac1\zeta.$$
	This is a contradiction; therefore, $u$ and $v$ are indeed in the same connected component of $G_+$.
	
	We can now apply Lemma~\ref{lem:separated-eff-res}:
	$$\frac1{R_{G_+}(u,v)}\ge\frac1{R_G(u,v)} - \zeta\ge\frac45\cdot\frac1{R_G(u,v)},$$
	as claimed.
\end{proof}

We are now ready to analyze Algorithm~\ref{alg:separated-approx-balanced}.

\begin{theorem}\label{thm:separated-approx-balanced}
Let $H=(V,E,w)$ be an $(\alpha,\beta)$-separated weighted hypergraph with $\beta\ge5|E|\cdot\gamma/\eta$.
Let $M=w(E)$ be the total hyperedge weight in $H$,
and let $w_{\min}$ the minimum hyperedge weight.
Then for $\gamma\ge4$, Algorithm~\ref{alg:separated-approx-balanced} terminates \emph{within $\alpha/\eta\cdot|E|^2\cdot\poly{n\log(M/w_{\min})}$ rounds} and returns an $\eta$-approximate $2\gamma$-balanced weight assignment for $H$.
\end{theorem}

\begin{proof}
  We call executions of the while-loop rounds and executions of the for-loop phases. It will be useful to denote $w_{i,\min} := \min_{e\in E_i}w(e)$ and $w_{i,\max} := \max_{e\in E_i}w(e)$.
  The definition of separated hypergraphs (Definition~\ref{def:separated}) then ensures that $w_{i,\max}\le\alpha\cdot w_{i,\min}$ and $w_{i,\min}\ge\beta\cdot w_{i+1,\max}$.

\paragraph{Algorithm Correctness.}
We first prove that upon termination, Algorithm~\ref{alg:separated-approx-balanced} indeed returns an $\eta$-approximate $2\gamma$-balanced weight assignment.
Consider a hyperedge $e$ and edges $f,g\in F_e$, with $\w(g)\geq \eta\cdot w(e)$. We will show that $R_{G^*}(f)\le2\gamma\cdot R_{G^*}(g)$, where $G^*$ is the final weight assignment returned by the algorithm at Line~\ref{line:separated-return}. This is exactly what is required by the definition of approximate balancedness, i.e. Definition~\ref{def:weighted-approx-balanced}.
Suppose $e\in E_i$.
Note that by the condition of the while-loop, at the termination of the $i^\text{th}$ phase, $e$, $f$, and $g$ satisfied even the stronger condition of $\eta$-approximate $\gamma$-balancedness.
Let us denote by $G'$ the weight assignment graph at the end of phase $i$.

Let us partition the edges of $G^*$ and $G'$ based on the weight of the corresponding hyperedge.
Denote $E_{\le i} := \bigcup_{j=1}^i E_j$ and $E_{>i} := \bigcup_{j=i+1}^\ell E_j$.
Similarly, let $F_{\le i} := \bigcup_{e\in E_{\le i}}F_e$ and $F_{>i} := \bigcup_{e\in E_{>i}}F_e$.
Finally, let $G^*_{\le i}$ and $G^*_{>i}$ be $G^*$ restricted to $F_{\le i}$ and $F_{>i}$ respectively, and define $G'_{\le i}$ and $G'_{>i}$ similarly.
Note that $G^*_{\le i}=G'_{\le i}$, since the algorithm never alters the weight assignments of $E_{\le i}$ after phase $i$.

From here, we will use two applications of Corollary~\ref{cor:separated-eff-res}. First, set $G_+:=G_{\le i}'$, $G_-:=G_{>i}'$ (which makes $G=G'$), and $(u,v):=f$. Then $\zeta$ is the total weight of all edges in $G_{>i}'$, which is the total weight of all hyperedges in $E_{>i}$, which is at most $w_{i+1,\max}\cdot|E|$. We verify the condition on $R_{G'}(f)$ from Corollary~\ref{cor:separated-eff-res}: $R_{G'}(f)\le\gamma\cdot R_{G'}(g)$, by the approximate $\gamma$-balancedness condition on $e$ in $G'$. Since $\w(g)\ge\eta\cdot w(e)$ (in both $G^*$ and $G'$), we have that $\w(g)\ge\eta\cdot w_{i,\min}$, and hence by Fact~\ref{fact:w-R} $R_{G'}(g)\le1/(\eta\cdot w_{i,\min})$. Finally, putting these together, as well as using the $(\alpha,\beta)$-separated quality of $H$, we get
$$R_{G'}(f)\le\frac{\gamma}{\eta\cdot w_{i,\min}}\le\frac{\gamma}{\eta\beta\cdot w_{i+1,\max}}\le\frac1{5\zeta},$$
by assumption on $\beta$. We can indeed apply Corollary~\ref{cor:separated-eff-res}, which gives us
\begin{equation}\label{eq:rgf-bound}
	R_{G'}(f)\ge\frac45\cdot R_{G_{\le i}}(f).
\end{equation}

We apply Corollary~\ref{cor:separated-eff-res} again, this time setting $G_+:=G^*_{\le i}$, $G_-:=G^*_{>i}$ (which makes $G=G^*$), and $(u,v):=g$. Then $\zeta\le w_{i+1,\max}\cdot|E|$ by an identical argument to the previous setting. We verify the condition on $R_{G^*}(g)$ from Corollary~\ref{cor:separated-eff-res}: $R_{G^*}(g)\le1/(\eta\cdot w_{i,\min})$ by Fact~\ref{fact:w-R}, since $\w(g)\ge\eta\cdot w(e)\ge\eta\cdot w_{i,\min}$ by assumption (in both $G^*$ and $G'$). Therefore, by an identical derivation to the previous case $R_{G^*}(g)\le1/(5\zeta)$. We can indeed apply Corollary~\ref{cor:separated-eff-res}, which gives us
\begin{equation}\label{eq:rgg-bound}
	R_{G^*}(g)\ge\frac45\cdot R_{G_{\le i}}(g).
\end{equation}

Putting together Equations~\eqref{eq:rgf-bound} and~\eqref{eq:rgg-bound}, along with the fact that $f$ and $g$ satisfied the condition of approximate $\gamma$-balancedness at the time of $G'$ (that is at the end of phase $i$), we get
\begin{align*}
R_{G^*}(f)
&\le R_{G_{\le i}}(f)
\le\frac54\cdot R_{G'}(f)
\le\frac{5\gamma}4\cdot R_{G'}(g)
\le\frac{5\gamma}{4}\cdot R_{G_{\le i}}(g)
\le\frac{25\gamma}{16}\cdot R_{G^*}(g)
\le2\gamma\cdot R_{G^*}(g),
\end{align*}
which concludes the proof that $G^*$ is $\eta$-approximately $2\gamma$-balanced.

\paragraph{Running Time.}
Next, we prove that Algorithm~\ref{alg:separated-approx-balanced} terminates within $\alpha/\eta\cdot|E|^2\cdot\poly{n\log(M/w_{\min})}$ rounds.
Similarly to the proof of Theorem~\ref{thm:greedy-approx-balanced}, we can bound the total growth of the ST-potential over the course of the algorithm by defining $G_0$ and $G_\infty$. Specifically, if $G_0$ is the starting graph of Algorithm~\ref{alg:separated-approx-balanced} and $G_\infty$ is the complete graph of uniform edge weight $M$,
then as before,
$$\Psi(G_\infty)-\Psi(G_0)\le\poly{n\log(M/w_{\min})}.$$

	We will now prove that the total number of rounds in the $i^\text{th}$ phase is at most $|E_i|\cdot\alpha/\eta\cdot|E|\cdot\poly{n\log(M/w_{\min})}$. Due to Lemma~\ref{lem:greedy-step}, the contribution of each update to the potential $\Psi(G)$ is at least
	$$\log\left(1+\lambda\gamma\cdot R_G(g)-\lambda\cdot R_G(g)-\lambda^2\gamma\cdot R_G(g)^2\right).$$ Similarly to the proof of Theorem~\ref{thm:greedy-approx-balanced}, this either means that $\lambda=1/(2R_G(g))$ and therefore the increase to $\Psi$ is at least $\log(5/4)$, or that $\lambda=\w(g)\le1/(2R_G(g))$, in which case the increase to $\Psi$ is at least $\log(1+\w(g)\cdot R_G(g))$. In the latter case, we further distinguish based on the value of $R_G(g)$.

	\begin{enumerate}
		\item $\lambda=1/(2R_G(g))$. In this case the ST-potential increases by at least $\log(5/4)$, and so there can be at most $\poly{n\log(M/w_{\min})}$ such rounds.
		\item $\lambda=\w(g)$ and $R_G(g)>1/(5w_{i,\max}\cdot|E|)$. In this case the ST-potential increases by at least $\log(1+\w(g)\cdot R_G(g))$. Since $g$ was selected in this round, $\w(g) \geq \eta\cdot w(e)\ge\eta\cdot w_{i,\min}$, and therefore the increase in $\Psi(G)$ is at least
		$$\log(1+\w(g)\cdot R_G(g))\ge\log\left(1+\frac{\eta\cdot w_{i,\min}}{5w_{i,\max}\cdot|E|}\right)\ge\log\left(1+\frac{\eta}{5\alpha\cdot|E|}\right)\ge\frac{\eta}{10\alpha\cdot|E|}.$$
		Thus, there can be at most $\alpha/\eta\cdot|E|\cdot\poly{n\log(M/w_{\min})}$ such rounds.
		\item $\lambda=\w(g)$ and $R_G(g)\le1/(5w_{i,\max}\cdot|E|)$. This is the most complicated case to analyze and we will be focusing on this for the rest of the proof.
	\end{enumerate}

	Suppose we are in the $i^\text{th}$ phase of Algorithm~\ref{alg:separated-approx-balanced} and we have $e\in E_i$ and $f,g\in F_e$ such that $R_G(f)>\gamma\cdot R_G(g)\ge4\cdot R_G(g)$. Further suppose that $R_G(g)\le1/(5w_{i,\max}\cdot|E|)$. Similarly to before, we divide $G$ into two graphs: $G_{<i}$ consisting of edges belonging to hyperedges from $E_j$ for $j<i$, and $G_{\ge i}$ consisting of edges belonging to hyperedges from $E_j$ for $j\ge i$. (Note that earlier in this proof we used the slightly different split into $G_{\le i}$ and $G_{>i}$.) We can again use Corollary~\ref{cor:separated-eff-res} to relate $R_G$ to $R_{G_{<i}}$.

	We apply Corollary~\ref{cor:separated-eff-res}, setting $G_+:=G_{<i}$, $G_-:=G_{\ge i}$ (which makes $G$ from Corollary~\ref{cor:separated-eff-res} the current graph $G$), and $(u,v):=g$. Then $\zeta$ is the total weight of all edges in $G_{\ge i}$, which is the total weight of all hyperedges in $E_{\ge i}$, which is at most $w_{i,\max}\cdot|E|$. Our above assumption $R_G(g)\le1/(5w_{i,\max}\cdot|E|)$ exactly guarantees that $R_G(g)\le1/(5\zeta)$ is satisfied and we can therefore apply Corollary~\ref{cor:separated-eff-res}:
	$$R_{G_{<i}}(g)\le\frac54\cdot R_{G}(g)\le\frac{5}{4\gamma}R_G(f)\le\frac{5}{16}R_G(f)\le\frac5{16}R_{G_{<i}}(g)<R_{G_{<i}}(f).$$

	In words, such weight transfers are always from edges of lower $G_{<i}$-resistance to those of strictly higher $G_{<i}$-resistance---a metric that never changes, since $G_{<i}$ is unchanged during the $i^\text{th}$ stage of the algorithm. Therefore, there cannot be $n^4\cdot|E_i|$ consecutive steps of type 3 in phase $i$. We prove this formally in the following claim.
	\begin{claim}\label{claim:type-three}
		There cannot be $n^4\cdot|E_i|$ consecutive rounds of type 3 in phase $i$.
	\end{claim}
\begin{proof}
	Suppose for contradiction that there is a sequence of $n^4\cdot|E_i|$ rounds of type 3. Then there must be a hyperedge $e\in E_i$ for which at least $n^4$ of these updates take place in $F_e$. We know that updates can only shift weight from $g$ to $f$ where $R_{G_{<i}}(f)>R_{G_{<i}}(g)$. Therefore, let us order the edges of $F_e$ by $R_{G_{<i}}$, that is, let $\rho:F_e\to\mathbb N$ be such that the $j^\text{th}$ largest edge in terms of $R_{G_{<i}}$, say $f$, has $\rho(f)=j$ (breaking ties arbitrarily). Then we can define a local potential function for $e$:
	$$\psi_e(G)=\sum_{f\in F_e}\rho(f)\cdot\mathbbm1(\w(f)>0).$$
	Note that $\psi_e$ starts out as at most $1 + 2 + \cdots + |F_e|\le n^4$, and after every update of type 3 it decreases by at least one.
        Since we have no updates of any type other than 3 (by assumption),
        and updates to other hyperedges do not affect $\psi_e$,
        we arrive at the contradiction that $\psi_e$ must become negative.
\end{proof}

	As a result of Claim~\ref{claim:type-three}, every consecutive sequence of $n^4\cdot|E_i|$ rounds must contain an update of type 1 or 2. Since we showed that there are at most $\alpha/\eta\cdot|E|\cdot\poly{n\log(M/w_{\min})}$ such updates, there can be at most $|E_i|\cdot\alpha/\eta\cdot|E|\cdot\poly{n\log(M/w_{\min})}$ rounds in phase $i$. Summing this over all phases, we get that there can be at most $\alpha/\eta\cdot|E|^2\cdot\poly{n\log(M/w_{\min})}$ rounds overall, as claimed.
\end{proof}

\subsection{Proof of the Second Part of Theorem~\ref{thm:main}}

We are now ready to prove the second part of Theorem~\ref{thm:main}.
Specifically, we prove that \textsc{FastSparsification} (Algorithm~\ref{alg:final} below) provides the result of Theorem~\ref{thm:main}:

\begin{algorithm}[H]
	\caption{Algorithm constructing nearly linear-sized spectral sparsifier, in nearly linear time.}\label{alg:final}
	\begin{algorithmic}[1]
		\Procedure{FastSparsification}{$H=(V,E,w)$}
		\State $H'=(V,E',w')\gets\Call{PolynomialSizeSparsification}{H,\epsilon/3}$\Comment{Algorithm~\ref{alg:poly-sparsification}}
		\For{$i$ from $1$ to $n$}
			\State $E_i'\gets\{e\in E_i' \mid w'(e)\in[n^{10(i-1)},n^{10i})\}$\label{line:E-i-def}
		\EndFor
		\State $H'_1\gets(V,\bigcup_{j \in \{1,\ldots,n\}, j\text{~odd}} E_i,w')$\label{line:H1-def}
		\State $H'_2\gets(V,\bigcup_{j \in \{1,\ldots,n\}, j\text{~even}}E_i,w')$\label{line:H2-def}
		\State $G_1\gets\Call{SeparatedApproxBalancing}{H_1',(E_j)_{j\text{~odd}},n^{10},n^{10},4,1/n^2}$\Comment{Algorithm~\ref{alg:separated-approx-balanced}}
		\State $G_2\gets\Call{SeparatedApproxBalancing}{H_2',(E_j)_{j\text{~even}},n^{10},n^{10},4,1/n^2}$
		\State $\lambda \gets \Theta(\epsilon^{-4}\log^3 n)$
		\State $\wt H_1\gets\Call{Sparsification}{H_1',G_1,\epsilon/3,\lambda}$ \Comment{Algorithm~\ref{alg:main}}
		\State $\wt H_2\gets\Call{Sparsification}{H_2',G_2,\epsilon/3,\lambda}$
		\State \Return $\wt H\gets\wt H_1\cup\wt H_2$
		\EndProcedure
	\end{algorithmic}
\end{algorithm}
In the final line $\wt H\gets\wt H_1\cup\wt H_2$ means that we take the union of hyperedges and weight functions from $\wt H_1$ and $\wt H_2$, since both are on the same vertex set $V$.


\begin{proofof}{the second part of Theorem~\ref{thm:main}}
	Indeed, Algorithm~\ref{alg:final} does exactly that. 
	
	First note a few observations about the steps of Algorithm~\ref{alg:final}: Notice that \Call{PolynomialSizeSparsification}{} produces hyperedge weights only in the range $[1,n^{10n})$, so $E'$ is partitioned into $E_i'$ for $i$ from $1$ to $n$. Next, $H_1'$ and $H_2'$ are indeed $(n^{10},n^{10})$-separated hypergraphs by the definitions in Lines~\ref{line:E-i-def},~\ref{line:H1-def}, and~\ref{line:H2-def}. Furthermore, note that $\beta\ge5|E|\cdot\gamma/\eta$ holds (as required by Theorem~\ref{thm:separated-approx-balanced}), since $\beta=n^{10}$, $\eta=1/n^2$, and $|E|=o(n^8)$ by the size guarantee of Corollary~\ref{cor:poly-sparsification}. Finally, $G_1$ and $G_2$ are $1/n^2$-approximate $8$-balanced weight assignments of $H_1'$ and $H_2'$ respectively, by Theorem~\ref{thm:separated-approx-balanced}.
	
	Therefore, by Lemma~\ref{lem:main}, $\wt H_1$ and $\wt H_2$ are $\epsilon/3$-spectral sparsifiers of $H_1'$ and $H_2'$ respectively. By the additivity of the hypergraph quadratic form $\wt H$ is an $\epsilon/3$-spectral sparsifier of $H'$. Since $H'$ is itself an $\epsilon/3$-spectral sparsifier of $H$, this means that $\wt H$ is an $\epsilon$-spectral sparsifier of $H$, as claimed.
	
	Moreover, by Lemma~\ref{lem:hypergraph_size}, $\wt H_1$ and $\wt H_2$ are both of size at most $O(n\log^3(n)/\epsilon^4)$ (since we set $\lambda$ to be $\Theta(\log^3(n)/\epsilon^4)$), and therefore so is $\wt H$.
	
	Finally, Algorithm~\ref{alg:final} runs in time $\wt O(mr+\poly{n})$. Indeed \Call{PolynomialSizeSparsification}{} runs in time $\wt O(mr) + \poly{n}$, as shown in Lemma~\ref{cor:poly-sparsification}; \Call{SeparatedApproxBalancing}{} runs in time $\poly{n}$ by Theorem~\ref{thm:separated-approx-balanced}, since $\alpha=\poly{n}$ and $M=\exp(O(n))$; \Call{Sparsification}{} runs in time $\poly{n}$, since the input hypergraph is of $\poly{n}$ size.
\end{proofof}


%% file: 000-main.bbl
\newcommand{\etalchar}[1]{$^{#1}$}
\begin{thebibliography}{YNY{\etalchar{+}}19}

\bibitem[ACK{\etalchar{+}}16]{ACKQWZ16}
Alexandr Andoni, Jiecao Chen, Robert Krauthgamer, Bo~Qin, David~P. Woodruff,
  and Qin Zhang.
\newblock On sketching quadratic forms.
\newblock In {\em Proceedings of the 2016 Conference on Innovations in
  Theoretical Computer Science (ITCS)}, pages 311--319, 2016.
\newblock \href {https://doi.org/10.1145/2840728.2840753}
  {\path{doi:10.1145/2840728.2840753}}.

\bibitem[AK17]{AK17}
N.~Alon and B.~Klartag.
\newblock Optimal compression of approximate inner products and dimension
  reduction.
\newblock In {\em Proceedings of the IEEE 58th Annual Symposium on Foundations
  of Computer Science (FOCS)}, pages 639--650, 2017.
\newblock \href {https://doi.org/10.1109/FOCS.2017.65}
  {\path{doi:10.1109/FOCS.2017.65}}.

\bibitem[AS08]{DBLP:books/daglib/0021015}
Noga Alon and Joel~H. Spencer.
\newblock {\em The Probabilistic Method, Third Edition}.
\newblock Wiley-Interscience series in discrete mathematics and optimization.
  Wiley, 2008.

\bibitem[BK15]{Benczur2015}
Andr{\'{a}}s~A. Bencz{\'{u}}r and David~R. Karger.
\newblock Randomized approximation schemes for cuts and flows in capacitated
  graphs.
\newblock {\em {SIAM} Journal on Computing}, 44(2):290--319, 2015.
\newblock \href {https://doi.org/10.1137/070705970}
  {\path{doi:10.1137/070705970}}.

\bibitem[BSS12]{BatsonSS12}
Joshua~D. Batson, Daniel~A. Spielman, and Nikhil Srivastava.
\newblock Twice-{Ramanujan} sparsifiers.
\newblock {\em SIAM Journal on Computing}, 41(6):1704--1721, 2012.
\newblock \href {https://doi.org/10.1137/090772873}
  {\path{doi:10.1137/090772873}}.

\bibitem[BST19]{Bansal2019}
Nikhil Bansal, Ola Svensson, and Luca Trevisan.
\newblock New notions and constructions of sparsification for graphs and
  hypergraphs.
\newblock In {\em Proceedings of the {IEEE} 60th Annual Symposium on
  Foundations of Computer Science (FOCS)}, pages 910--928, 2019.
\newblock \href {https://doi.org/10.1109/focs.2019.00059}
  {\path{doi:10.1109/focs.2019.00059}}.

\bibitem[CKN20]{Chen2020}
Yu~Chen, Sanjeev Khanna, and Ansh Nagda.
\newblock Near-linear size hypergraph cut sparsifiers.
\newblock In {\em Proceedings of the 61st {IEEE} Annual Symposium on
  Foundations of Computer Science, {FOCS} 2020}, pages 61--72, 2020.
\newblock \href {https://doi.org/10.1109/FOCS46700.2020.00015}
  {\path{doi:10.1109/FOCS46700.2020.00015}}.

\bibitem[CKST19]{CKST19}
Charles Carlson, Alexandra Kolla, Nikhil Srivastava, and Luca Trevisan.
\newblock Optimal lower bounds for sketching graph cuts.
\newblock In {\em Proceedings of the 13th Annual ACM-SIAM Symposium on Discrete
  Algorithms (SODA)}, pages 2565--2569, 2019.

\bibitem[CLTZ18]{Chan2018}
T.-H.~Hubert Chan, Anand Louis, Zhihao~Gavin Tang, and Chenzi Zhang.
\newblock Spectral properties of hypergraph {Laplacian} and approximation
  algorithms.
\newblock {\em Journal of the {ACM}}, 65(3):1--48, 2018.
\newblock \href {https://doi.org/10.1145/3178123} {\path{doi:10.1145/3178123}}.

\bibitem[FSY18]{fujii2018polynomialtime}
Kaito Fujii, Tasuku Soma, and Yuichi Yoshida.
\newblock Polynomial-time algorithms for submodular {Laplacian} systems, 2018.
\newblock \href {http://arxiv.org/abs/1803.10923} {\path{arXiv:1803.10923}}.

\bibitem[HSJR13]{Hein2013}
Matthias Hein, Simon Setzer, Leonardo Jost, and Syama~Sundar Rangapuram.
\newblock The total variation on hypergraphs - learning on hypergraphs
  revisited.
\newblock In {\em Advances in Neural Information Processing Systems 26 (NIPS)},
  pages 2427--2435, 2013.

\bibitem[IKT21]{ikeda2021coarse}
Masahiro Ikeda, Yu~Kitabeppu, and Yuuki Takai.
\newblock Coarse {Ricci} curvature of hypergraphs and its generalization, 2021.
\newblock \href {http://arxiv.org/abs/2102.00698} {\path{arXiv:2102.00698}}.

\bibitem[IMTY19]{ikeda2019finding}
Masahiro Ikeda, Atsushi Miyauchi, Yuuki Takai, and Yuichi Yoshida.
\newblock Finding {Cheeger} cuts in hypergraphs via heat equation, 2019.
\newblock \href {http://arxiv.org/abs/1809.04396} {\path{arXiv:1809.04396}}.

\bibitem[KK15]{Kogan2015}
Dmitry Kogan and Robert Krauthgamer.
\newblock Sketching cuts in graphs and hypergraphs.
\newblock In {\em Proceedings of the 2015 Conference on Innovations in
  Theoretical Computer Science (ITCS)}, pages 367--376, 2015.
\newblock \href {https://doi.org/10.1145/2688073.2688093}
  {\path{doi:10.1145/2688073.2688093}}.

\bibitem[KKTY21]{KKTY21}
Michael Kapralov, Robert Krauthgamer, Jakab Tardos, and Yuichi Yoshida.
\newblock Towards tight bounds for spectral sparsification of hypergraphs.
\newblock In {\em Proceedings of the 53rd Annual ACM Symposium on Theory of
  Computing (STOC)}, 2021.
\newblock to appear.

\bibitem[Kom67]{komura1967nonlinear}
Yukio Komura.
\newblock Nonlinear semi-groups in {Hilbert} space.
\newblock {\em Journal of the Mathematical Society of Japan}, 19(4):493--507,
  1967.

\bibitem[LM18]{Li2018}
Pan Li and Olgica Milenkovic.
\newblock Submodular hypergraphs: {$p$}-{Laplacians}, {Cheeger} inequalities
  and spectral clustering.
\newblock In {\em Proceedings of the 35th International Conference on Machine
  Learning (ICML)}, volume~80, pages 3020--3029, 2018.

\bibitem[Lou15]{Louis2015}
Anand Louis.
\newblock Hypergraph {Markov} operators, eigenvalues and approximation
  algorithms.
\newblock In {\em Proceedings of the 47th Annual {ACM} on Symposium on Theory
  of Computing (STOC)}, pages 713--722, 2015.
\newblock \href {https://doi.org/10.1145/2746539.2746555}
  {\path{doi:10.1145/2746539.2746555}}.

\bibitem[Lov93]{lovasz1993random}
L{\'a}szl{\'o} Lov{\'a}sz.
\newblock Random walks on graphs: A survey.
\newblock In D.~{Mikl\'os}, Vera~T. {S\'os}, and T.~{Sz\"onyi}, editors, {\em
  Combinatorics, Paul Erd{\H o}s is eighty. Vol. 2}, pages 353--397. J\'anos
  Bolyai Mathematical Society, 1993.
\newblock URL: \url{http://bolyai.math.elte.hu/~lovasz/erdos.pdf}.

\bibitem[LS15]{LeeS15a}
Yin~Tat Lee and He~Sun.
\newblock Constructing linear-sized spectral sparsification in almost-linear
  time.
\newblock In Venkatesan Guruswami, editor, {\em {IEEE} 56th Annual Symposium on
  Foundations of Computer Science, {FOCS} 2015, Berkeley, CA, USA, 17-20
  October, 2015}, pages 250--269. {IEEE} Computer Society, 2015.
\newblock \href {https://doi.org/10.1109/FOCS.2015.24}
  {\path{doi:10.1109/FOCS.2015.24}}.

\bibitem[LS17]{Lee017}
Yin~Tat Lee and He~Sun.
\newblock An sdp-based algorithm for linear-sized spectral sparsification.
\newblock In Hamed Hatami, Pierre McKenzie, and Valerie King, editors, {\em
  Proceedings of the 49th Annual {ACM} {SIGACT} Symposium on Theory of
  Computing, {STOC} 2017, Montreal, QC, Canada, June 19-23, 2017}, pages
  678--687. {ACM}, 2017.
\newblock \href {https://doi.org/10.1145/3055399.3055477}
  {\path{doi:10.1145/3055399.3055477}}.

\bibitem[Mey73]{Meyer73}
C.~Meyer.
\newblock Generalized inversion of modified matrices.
\newblock {\em {SIAM} Journal on Applied Mathematics}, 24(3):315--323, 1973.

\bibitem[Miy92]{miyadera1992nonlinear}
Isao Miyadera.
\newblock {\em Nonlinear Semigroups}, volume 109.
\newblock American Mathematical Soc., 1992.

\bibitem[NR13]{Newman2013}
Ilan Newman and Yuri Rabinovich.
\newblock On multiplicative {$\lambda$}-approximations and some geometric
  applications.
\newblock {\em {SIAM} Journal on Computing}, 42(3):855--883, 2013.
\newblock \href {https://doi.org/10.1137/100801809}
  {\path{doi:10.1137/100801809}}.

\bibitem[RV07]{RudelsonV07}
Mark Rudelson and Roman Vershynin.
\newblock Sampling from large matrices: An approach through geometric
  functional analysis.
\newblock {\em Journal of the ACM}, 54(4):21, 2007.
\newblock \href {https://doi.org/10.1145/1255443.1255449}
  {\path{doi:10.1145/1255443.1255449}}.

\bibitem[SS11]{Spielman2011a}
Daniel~A. Spielman and Nikhil Srivastava.
\newblock Graph sparsification by effective resistances.
\newblock {\em {SIAM} Journal on Computing}, 40(6):1913--1926, 2011.
\newblock \href {https://doi.org/10.1137/080734029}
  {\path{doi:10.1137/080734029}}.

\bibitem[ST11]{Spielman2011}
Daniel~A. Spielman and Shang-Hua Teng.
\newblock Spectral sparsification of graphs.
\newblock {\em {SIAM} Journal on Computing}, 40(4):981--1025, 2011.
\newblock \href {https://doi.org/10.1137/08074489x}
  {\path{doi:10.1137/08074489x}}.

\bibitem[SY19]{Soma2019}
Tasuku Soma and Yuichi Yoshida.
\newblock Spectral sparsification of hypergraphs.
\newblock In {\em Proceedings of the 30th Annual {ACM}-{SIAM} Symposium on
  Discrete Algorithms (SODA)}, pages 2570--2581. jan 2019.
\newblock \href {https://doi.org/10.1137/1.9781611975482.159}
  {\path{doi:10.1137/1.9781611975482.159}}.

\bibitem[Ten16]{Teng2016}
Shang-Hua Teng.
\newblock Scalable algorithms for data and network analysis.
\newblock {\em Foundations and Trends{\textregistered} in Theoretical Computer
  Science}, 12(1-2):1--274, 2016.
\newblock \href {https://doi.org/10.1561/0400000051}
  {\path{doi:10.1561/0400000051}}.

\bibitem[TMIY20]{Takai2020}
Yuuki Takai, Atsushi Miyauchi, Masahiro Ikeda, and Yuichi Yoshida.
\newblock Hypergraph clustering based on {PageRank}.
\newblock In {\em Proceedings of the 26th {ACM} {SIGKDD} International
  Conference on Knowledge Discovery {\&} Data Mining (KDD)}, pages 1970--1978,
  2020.
\newblock \href {https://doi.org/10.1145/3394486.3403248}
  {\path{doi:10.1145/3394486.3403248}}.

\bibitem[Tro11]{Tropp2011}
Joel~A. Tropp.
\newblock User-friendly tail bounds for sums of random matrices.
\newblock {\em Foundations of Computational Mathematics}, 12(4):389--434, 2011.
\newblock \href {https://doi.org/10.1007/s10208-011-9099-z}
  {\path{doi:10.1007/s10208-011-9099-z}}.

\bibitem[Vis13]{Vishnoi2013}
Nisheeth~K. Vishnoi.
\newblock Lx = b.
\newblock {\em Foundations and Trends{\textregistered} in Theoretical Computer
  Science}, 8(1{\textendash}2):1--141, 2013.
\newblock \href {https://doi.org/10.1561/0400000054}
  {\path{doi:10.1561/0400000054}}.

\bibitem[YNN{\etalchar{+}}20]{Yadati2020}
Naganand Yadati, Vikram Nitin, Madhav Nimishakavi, Prateek Yadav, Anand Louis,
  and Partha Talukdar.
\newblock {NHP}: Neural hypergraph link prediction.
\newblock In {\em Proceedings of the 29th {ACM} International Conference on
  Information {\&} Knowledge Management (CIKM)}, pages 1705--1714, 2020.
\newblock \href {https://doi.org/10.1145/3340531.3411870}
  {\path{doi:10.1145/3340531.3411870}}.

\bibitem[YNY{\etalchar{+}}19]{Yadati2019}
Naganand Yadati, Madhav Nimishakavi, Prateek Yadav, Vikram Nitin, Anand Louis,
  and Partha~P. Talukdar.
\newblock {HyperGCN}: {A} new method for training graph convolutional networks
  on hypergraphs.
\newblock In {\em Advances in Neural Information Processing Systems 32}, pages
  1509--1520, 2019.

\bibitem[Yos16]{Yoshida2016}
Yuichi Yoshida.
\newblock Nonlinear {L}aplacian for digraphs and its applications to network
  analysis.
\newblock In {\em Proceedings of the 9th {ACM} International Conference on Web
  Search and Data Mining (WSDM)}, pages 483--492, 2016.
\newblock \href {https://doi.org/10.1145/2835776.2835785}
  {\path{doi:10.1145/2835776.2835785}}.

\bibitem[Yos19]{Yoshida2019}
Yuichi Yoshida.
\newblock Cheeger inequalities for submodular transformations.
\newblock In {\em Proceedings of the 30th Annual {ACM}-{SIAM} Symposium on
  Discrete Algorithms (SODA)}, pages 2582--2601, 2019.
\newblock \href {https://doi.org/10.1137/1.9781611975482.160}
  {\path{doi:10.1137/1.9781611975482.160}}.

\bibitem[ZHTC20]{Zhang2020}
Chenzi Zhang, Shuguang Hu, Zhihao~Gavin Tang, and T-H.~Hubert Chan.
\newblock Re-revisiting learning on hypergraphs: Confidence interval,
  subgradient method, and extension to multiclass.
\newblock {\em {IEEE} Transactions on Knowledge and Data Engineering},
  32(3):506--518, 2020.
\newblock \href {https://doi.org/10.1109/tkde.2018.2880448}
  {\path{doi:10.1109/tkde.2018.2880448}}.

\bibitem[ZLO15]{ZhuLO15}
Zeyuan~Allen Zhu, Zhenyu Liao, and Lorenzo Orecchia.
\newblock Spectral sparsification and regret minimization beyond matrix
  multiplicative updates.
\newblock In Rocco~A. Servedio and Ronitt Rubinfeld, editors, {\em Proceedings
  of the Forty-Seventh Annual {ACM} on Symposium on Theory of Computing, {STOC}
  2015, Portland, OR, USA, June 14-17, 2015}, pages 237--245. {ACM}, 2015.
\newblock \href {https://doi.org/10.1145/2746539.2746610}
  {\path{doi:10.1145/2746539.2746610}}.

\end{thebibliography}
